\documentclass[times, doublespace]{wcmauth}
\usepackage{moreverb}

\usepackage{comment}
\usepackage{graphicx}
\usepackage{graphicx}
\usepackage{latexsym}
\usepackage{wrapfig}
\usepackage{amsmath}
\usepackage{amsthm}
\usepackage{amsfonts}
\usepackage{verbatim}
\usepackage{subfigure}
\usepackage{comment}
\usepackage{lipsum}
\usepackage{algpseudocode,algorithm}
\algrenewcommand\alglinenumber[1]{\tiny #1:}
\usepackage{url}
\usepackage[square,numbers]{natbib}

\newtheorem{theorem}{\textbf{Theorem}}

\newcommand\BibTeX{{\rmfamily B\kern-.05em \textsc{i\kern-.025em b}\kern-.08em
T\kern-.1667em\lower.7ex\hbox{E}\kern-.125emX}}

\begin{document}

\runningheads{B. S. Peres and O. Goussevskaia}{MHCL: IPv6 Multihop Host Configuration for Low-Power Wireless Networks}

\articletype{RESEARCH ARTICLE}

\title{MHCL: IPv6 Multihop Host Configuration for Low-Power Wireless Networks}

\author{Bruna S. Peres\corrauth ~and Olga Goussevskaia}

\address{Department of Computer Science, Universidade Federal de Minas Gerais (UFMG), Belo Horizonte - MG, Brazil}

\corraddr{Bruna S. Peres. DCC/ICEx/UFMG, Av. Antonio Carlos, 6627, Belo Horizonte, MG,
31270-901, Brazil. E-mail: bperes@dcc.ufmg.br}

\begin{abstract}
Standard routing protocols for Low power and Lossy Networks
are typically designed to optimize bottom-up data flows, by maintaining a
cycle-free network topology. The advantage of such topologies is low memory
footprint to store routing information (only the parent's address needs to me
known by each node). The disadvantage is that other communication
patterns, like top-down and bidirectional data flows, are not
easily implemented. In this work we propose MHCL: IPv6 Multihop Host
Configuration for Low-Power Wireless Networks. MHCL employs hierarchical
address allocation that explores cycle-free network topologies and aims to
enable top-down data communication with low message overhead and memory footprint.
We evaluated the performance of MHCL both analytically and through simulations.
We implemented MHCL as a subroutine of RPL protocol on Contiki OS and showed
that it significantly improves top-down message delivery in RPL, while using a
constant amount of memory (i.e., independent of network size) and being efficient in terms of setup time and number of control
messages.

\end{abstract}


\maketitle


\section{Introduction}

 
The main function of a low-power wireless network is usually some sort of data collection. Applications based on data collection are plentiful, examples include environment monitoring~\cite{Tolle:2005}, field surveillance~\cite{Vicaire:2009},
and scientific observation~\cite{Werner-Allen:2006}. In order to perform data collection, a cycle-free graph structure is typically maintained and a convergecast is implemented on this network topology.
Many operating systems for sensor nodes (e.g. Tiny OS~\cite{Levis2005} and Contiki OS~\cite{Dunkels:2004}) implement mechanisms (e.g. Collection Tree Protocol (CTP)~\cite{Fonseca:2009} or the IPv6 Routing Protocol for Low-Power and Lossy Networks (RPL)~\cite{rfc6550}) to maintain cycle-free network topologies to support data-collection applications.

In some situations, however, data flow in the opposite direction 
--- from the root, or the border router, towards the leaves
becomes necessary. These
situations might arise in network configuration routines, specific data queries, or applications that require reliable data transmissions with
acknowledgments. For example, in the context of Internet of Thinks, or smart
homes, one can imagine a low-power wireless network connecting the
appliances and other electric devices of a home to a gateway, or a border
router. While away from home, one might wish to connect to the 
refrigerator at home to check whether one is out of milk and should buy some on
the way home. In order to deliver this request, a message would need to be sent
to the IPv6 address of the refrigerator. This message would first be delivered
to the home gateway and then be routed downwards through the multihop wireless
network connecting the appliances. Each node in this network would act as a
router and decide which way to forward the message, so that it reaches the refrigerator.

Even though top-down data traffic is typically not the main target of low-power multihop wireless networks, it is enabled by some popular routing
protocols in an alternative operating mode, so that routes are optimized
for bottom-up traffic, but top-down traffic is still possible.
In the RPL protocol, for example, nodes
maintain a DAG (Directed Acyclic Graph) topology, in which each node keeps in memory a small list of parent nodes, to which it forwards data upwards to the root. 
If a node wants to act not just as a
source but a destination, it sends a special message upwards through the DAG,
and the intermediate nodes (if they operate in the so-called storing mode) add
an entry for this destination in the routing table created specifically for
downward routing. The size of each such table is potentially $O(n)$, where $n$
is the size of the subtree rooted at the routing node. 
Given that memory space may be highly constrained in low-power
wireless networks, all routes often cannot be stored, and packets must therefore
be dropped. This results in high message loss and high memory footprint. In
order to reduce the size of the routing table, it would be desirable to
aggregate several addresses of closely located destination nodes in a single
routing table entry. However, because IPv6 addresses typically have their last
64 bits derived from the unique MAC address of each node, they are basically random and contain no information about the network topology.

In this work we propose MHCL: a Multihop Host Configuration strategy that 
explores cycle-free network structures in Low-power wireless networks to
generate and assign IPv6 addresses to nodes. The objective is to enable 
efficient and robust top-down data traffic with low memory footprint, i.e., small routing tables. 
We propose and analyze two strategies: Greedy (MHCL-G) and Aggregation (MHCL-A) address
allocation. In the MHCL-G approach, each node, starting with the root and
terminating with the leaves, partitions the address space available or assigned
to it by a parent (in case of the root, it could be the last 64 bits of the
gateway's IPv6 address) into equally-sized address ranges and assigns them
to its children (leaving a reserve address pool for possible future
connections).
The MHCL-A approach contains an initial aggregation phase, in which nodes
compute their number of descendants.\footnote{In case of DAG, where a node
might have multiple parents, the preferred parent is used for descendants
aggregation purposes.} Once the root receives the aggregated number of descendants of each
of its (immediate) children, it allocates an address range of size
proportional to the size of the subtree rooted at each child node (also
leaving a reserve address space for future connections). Once a node receives an
address range from its (preferred) parent, it partitions it among its own
children, until all nodes receive an IPv6 address.

MHCL partitions the available IPv6 address space hierarchically, based on the
topology of the underlying wireless network. As a result, routing tables for
top-down data traffic are of size $O(k)$, where $k$ is the number of (immediate)
children of each routing node in a cycle-free network topology. The protocol
uses Trickle-inspired~\cite{Levis:2004} timers to adjust the communication
routines and quickly adapt to dynamics in the network's topology but not flood
the network with too many control messages. We implemented MHCL as a subroutine
of RPL protocol in Contiki OS and analyzed its performance through simulations
using a network simulator provided by Contiki~\cite{Eriksson:2009}. Our results
indicate that MHCL significantly improves top-down message delivery of
RPL, while being efficient in terms of setup time and number of control
messages. We also evaluated how robust our strategy is to some network dynamics,
such as transient link and node failures.

The rest of this paper is organized as follows. In Section~\ref{sec:6lowpan}, we
briefly describe how address configuration and top-down data flows are
implemented in 6LoWPAN. In Section~\ref{sec:mhcl} we describe the MHCL protocol:
we define two IPv6 address partitioning strategies and describe the
communication routines and messages used.
In Section~\ref{sec:analysis} we analyze the time and message complexity of MHCL. In Section~\ref{sec:results}
 we present our experimental results. In Section~\ref{sec:related} we discuss related work, and in
Section~\ref{sec:conclusion} finish with some concluding remarks.

\section{6LoWPAN}
\label{sec:6lowpan}

6LoWPAN (IPv6 over Low power Wireless Personal Area Networks) are comprised of
low-cost wireless communication devices with limited resources, compatible with the IEEE 802.15.4 standard. In these networks, packet losses can be very frequent and links can become unusable for some time due to a variety of reasons. 
%
%

\textbf{IPv6 address assignment:} Between the (global) IPv6 network and a
6LoWPAN there is a border router, or gateway, that performs IPv6 header
compression. The 128 bits of an IPv6 address are divided into two parts: the
network prefix (64 bits) and the host address (64 bits). The 6LoWPAN header
compression mechanism omits the network prefix bits, since they are fixed for a
given 6LoWPAN. Since the remaining 64 bits can handle a very large address space
(up to $1.8 \times 10^{19}$), 6LoWPAN offers compression options for the host
address (typically only 16 bits are used). Similarly to IPv4, there are two ways
to configure an address: static and dynamic. Static address allocation has to be
configured manually at each device. Dynamic address allocation is performed by
DHCPv6 (Dynamic Host Configuration Protocol version 6). Since the network is multihop, this requires a DHCPv6
server and several retransmission agents to be deployed somewhere in the
network. Once the DHCPv6 server assigns a global address to a node, it starts
running the routing protocol (e.g. RPL).

\textbf{RPL (Acyclic Topology):} RPL~\cite{rfc6550} is a distance vector routing protocol specifically designed for 6LoWPAN. The foundation of the protocol is building and maintaining an acyclic network topology, directed at the root (the border router).  Therefore, a
Destination-Oriented Directed Acyclic Graph (DODAG) is created. The graph is constructed by the use of an Objective Function (OF) which defines
how the routing metric is computed (by using the expected number of transmissions (ETX) or the current amount of battery power of a node, for example). The root starts the DODAG construction by advertising messages of type DIO (DODAG Information Object). When a node decides to join the DODAG, it saves a list of potential parents, ordered according to respective link qualities, and sets the first parent in the list as its preferred parent. Then, the node computes its rank, which is a metric that indicates the coordinates of the node in the network hierarchy and is used to avoid routing loops. 

During the setup phase and in the event of a failure, nodes update their preferred parents. The Trickle algorithm~\cite{Levis:2004} is used to adapt the sending rate of messages to the network topology dynamics. In a network with stable links the control messages will be rare, whereas an environment in which the topology changes frequently will cause RPL to send control information more often. When this process stabilizes, the data collection can start.

\textbf{RPL (Downward Routing):} RPL specifies two modes of operation for downward routing: storing and non-storing. Both modes require messages of type DAO (Destination Advertisement Object) to be sent by nodes that wish to act as destinations. In the storing mode, this information is forwarded upward to preferred parents and each RPL router stores routes to its destinations in a downward routing table. The memory necessary to store such a table is $O(n)$, where $n$ is the size of the subtree rooted at the node. Since memory capacity is constrained in 6LoWPAN, often only a small portion of all routes can be stored in the routing tables, causing packets to be dropped. In the non-storing mode, nodes do not have storage capacity, so the DAO messages are sent directly to the root. The root is the only node that stores the downward routing information. To send messages, source routing is then used. Consequently, the network suffers greater fragmentation risk and data message loss, and the capacity of the network is excessively consumed by source routing~\cite{Clausen2013}.

\section{MHCL: IPv6 Multihop Host Configuration for Low-Power Wireless Networks}
\label{sec:mhcl}

The goal of MHCL (IPv6 Multihop Host Configuration for Low-power wireless
networks) is to implement an IPv6 address allocation scheme for downward
traffic that has low memory footprint, i.e.,
needs small routing tables.
The address space available to the border router, for instance the 64 least-significant bits
of the IPv6 address (or a compressed 16-bit representation of the latter), is
hierarchically partitioned among nodes connected to the border router through a
multihop cycle-free topology (implemented by standard protocols, such as
RPL~\cite{rfc6550} or CTP~\cite{Fonseca:2009}). Each node receives an address range from its parent and
partitions it among its children, until all nodes receive an address. Since the
address allocation is performed in a hierarchical way, the routing table
of each node can have $k$ entries, where $k$ is the number of its
(direct) children. Each routing table entry aggregates the addresses of all
nodes in the subtree rooted at the corresponding child-node.

In order to decide how the available address space is partitioned, nodes need to
collect information about the topology of the network. Once a \textit{stable}
view of the network's topology is achieved, the root starts distributing address
ranges downwards to all nodes. Note that the notion of stability is
important to implement a coherent address space partition. Therefore, MHCL has
an initial set-up phase, during which information about the topology is
progressively updated, until a (pre-defined) sufficiently long period of time
goes by without any topology changes. To implement this adaptive approach, we
use Trickle-inspired timers to trigger messages, as we explain in detail in
Section~\ref{subsec:timers}. Note that, once the network reaches an initial
state of stability, later changes to topology are expected to be of local
nature, caused by a link or a node failure, or a change in the preferred parent
of a node. In these cases, the address allocation does not need to be
updated, since local mechanisms of (best effort) message resubmission
can be used to improve message delivery rates.
This is confirmed by our experimental results in Section~\ref{sec:results}. In the (atypical)
event of non-local and permanent topology changes in the network, the MHCL
algorithm has to be restarted by the root.

In this section, we describe the main components of MHCL: the IPv6 address
space partitioning (Section~\ref{subsec:addr-space}), the message types
(Section~\ref{subsec:messages}), the communication routines and timers (Section~\ref{subsec:timers}), and the routing table and
forwarding routine (Section~\ref{subsec:routing}).

\subsection{IPv6 address space partition}\label{subsec:addr-space}

The routing operation in computer networks is strongly influenced by the network
topology and the semantics of the addresses. When the network is big enough to
make the hosts unable to maintain entries in a routing table for every router
and host in the network, usually, the addresses are hierarchically
structured~\cite{Tsuchiya92}. One way to encode hierarchical addresses is to
pre-assign positions and lengths to each field. The hierarchical addressing
allows prefix aggregation, which enables contiguous IPv6 addresses or prefixes
to be aggregated into a single prefix~\cite{6765943}. However, if the
hierarchical structure of the topology doesn't match the hierarchical structure
of the address, some addresses may be under-utilized (in a sparse network, for
example), while other may be overflowed. That address space insufficiency is the
main problem with fixed-position address fields. Because of this, instead of
working with prefixes, we decided to perform a hierarchical partition of the
available address space among nodes connected to the border router through a
multihop cycle-free topology. We propose two address space partitioning strategies: Greedy (MHCL-G) and Aggregate (MHCL-A).

\begin{figure}
\centering
\includegraphics[width=.7\columnwidth]{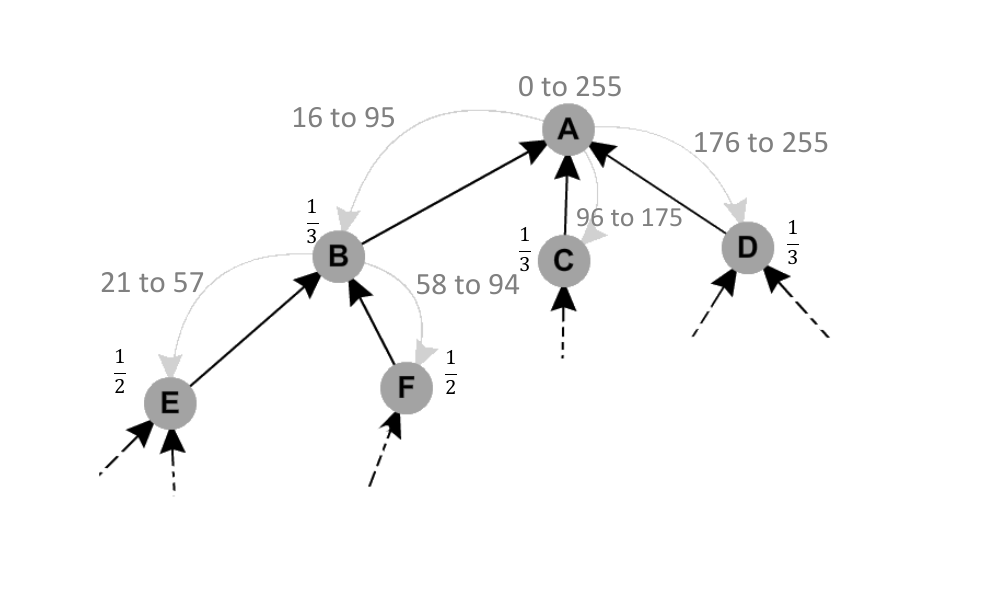}
            
						\caption{MHCL-G (Greedy address space partition). 8-bit address space
            at the root and $6.25\%$ reserve pool for future/delayed
            connections.}\label{fig:addrPartitionGreedy}
\end{figure}

\begin{figure}
\centering
\includegraphics[width=.7\columnwidth]{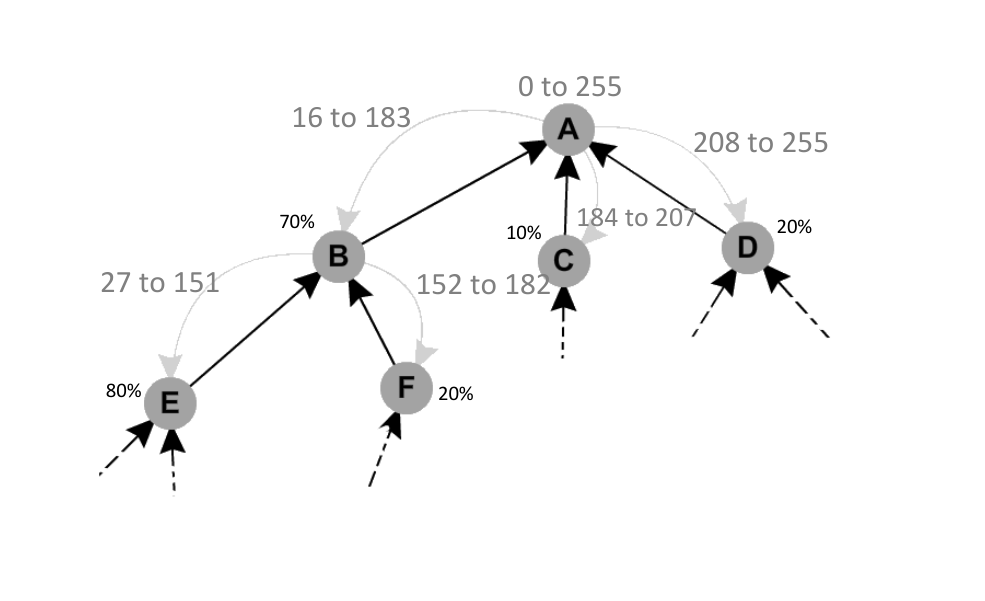}
\caption{MHCL-A (Aggregate address space partition). 8-bit address space
            at the root and $6.25\%$ reserve pool for future/delayed
            connections.}\label{fig:addrPartitionAggregate}
\end{figure}

In the greedy approach, each node $n_i$ counts the number $k$ of its (direct)
children (nodes whose preferred parent is $n_i$); waits until it is assigned an
address range by it parent (if $n_i$ is not the root); partitions the
address space available to it into $k$ address ranges of equal size, leaving a
reserve pool, say $r\%$ of its address space for possible future/delayed
connections (parameter $r$ can be configured according to the expected number
of newly deployed nodes in the network); and then distributes the resulting
address ranges to its children (see Figure~\ref{fig:addrPartitionGreedy}).
Note that MHCL-G uses only one-hop topology information, which allows it to have
a relatively fast set-up phase.

In the aggregate approach, each node $n_i$ counts the total number of its
descendants, i.e., the size of the subtree rooted at itself, and
propagates it to its parent.\footnote{In order to avoid double counting nodes
with multiple parents in a DAG structure, such as the one used in RPL, only the
preferred parent of each node is considered.} Moreover, $n_i$ saves the number
of descendants of each child. Once the root has received the (aggregate) number of
descendants of each child, it partitions the available address space into $k$
ranges of size proportional to the size of the subtree rooted at each child,
also leaving a portion of $r\%$ as reserve (see
Figure~\ref{fig:addrPartitionAggregate}). Each node $n_i$ repeats the space partitioning procedure upon receiving its own address space from the parent and
sends the proportional address ranges to the respective children. The idea is to
allocate larger portions to larger subtrees, which becomes important in especially large networks, because it maximizes the address space utilization.
Note that this approach needs information aggregated along multiple hops, which
results in a longer set-up phase.

\subsection{Messages}\label{subsec:messages}

MHCL uses four message types, which we implemented by modifying the DIO and DAO
message types defined in the RPL protocol: $DIO_{MHCL}$, $DIOACK_{MHCL}$,
$DAO_{MHCL}$, and $DAOACK_{MHCL}$.

Messages of type $DIO_{MHCL}$ are sent along downward routes, from parent to
child. This message is used for address allocation and contains the address and
corresponding address partition assigned to a child node by its parent. The
message is derived from the DIO message type used in RPL, with \textit{flag} field set to $1$ (in RPL the flag
is set to $0$) and child node's address (which is the first address in the
allocated IPv6 range) and the allocated address partition size sent in the
\textit{options} field (see Figure~\ref{fig:diomhcl}). Note that the size of the first address
and the size of the allocated address partition can have a length pre-defined by
the root, according to the overall address space (we used a value of 16 bits,
because the compressed host address has 16 bits).
This information is sufficient for the child node to decode the message and
execute the address allocation procedure for its children. Messages of type
$DIOACK_{MHCL}$ are identified by the \textit{flag} field set to $2$ and are
used to acknowledge the reception of an allocated address range through a
$DIO_{MHCL}$ message.
If no acknowledgment is received (after a timer of type $DIO_{min}$ goes off,
see Section~\ref{subsec:timers}), then the parent tries to retransmit the
message. 

\begin{figure}
\centering
\includegraphics[width=.7\columnwidth]{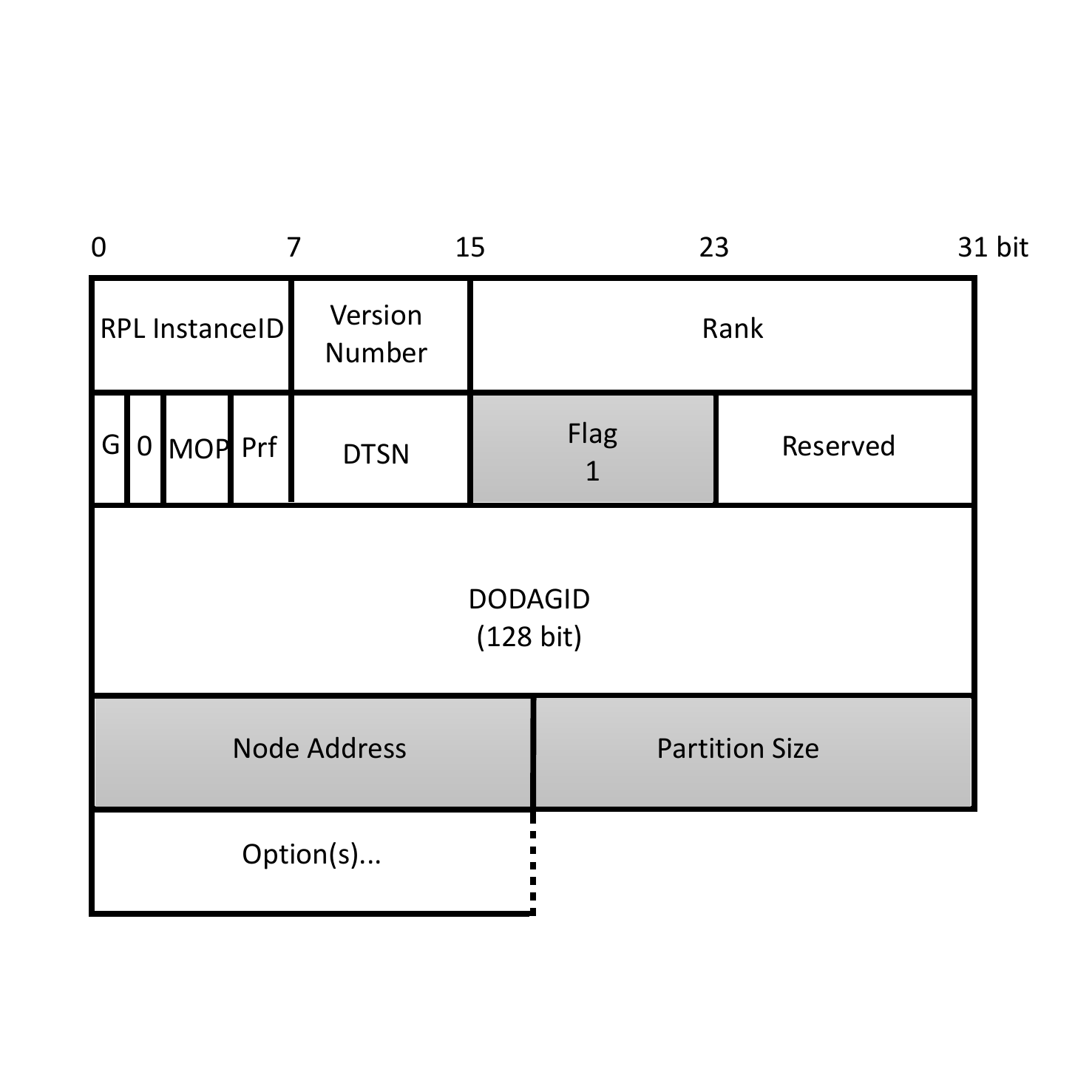}
						\caption{$DIO_{MHCL}$ message.}\label{fig:diomhcl}
\end{figure}

\begin{figure}
\centering
\includegraphics[width=.7\columnwidth]{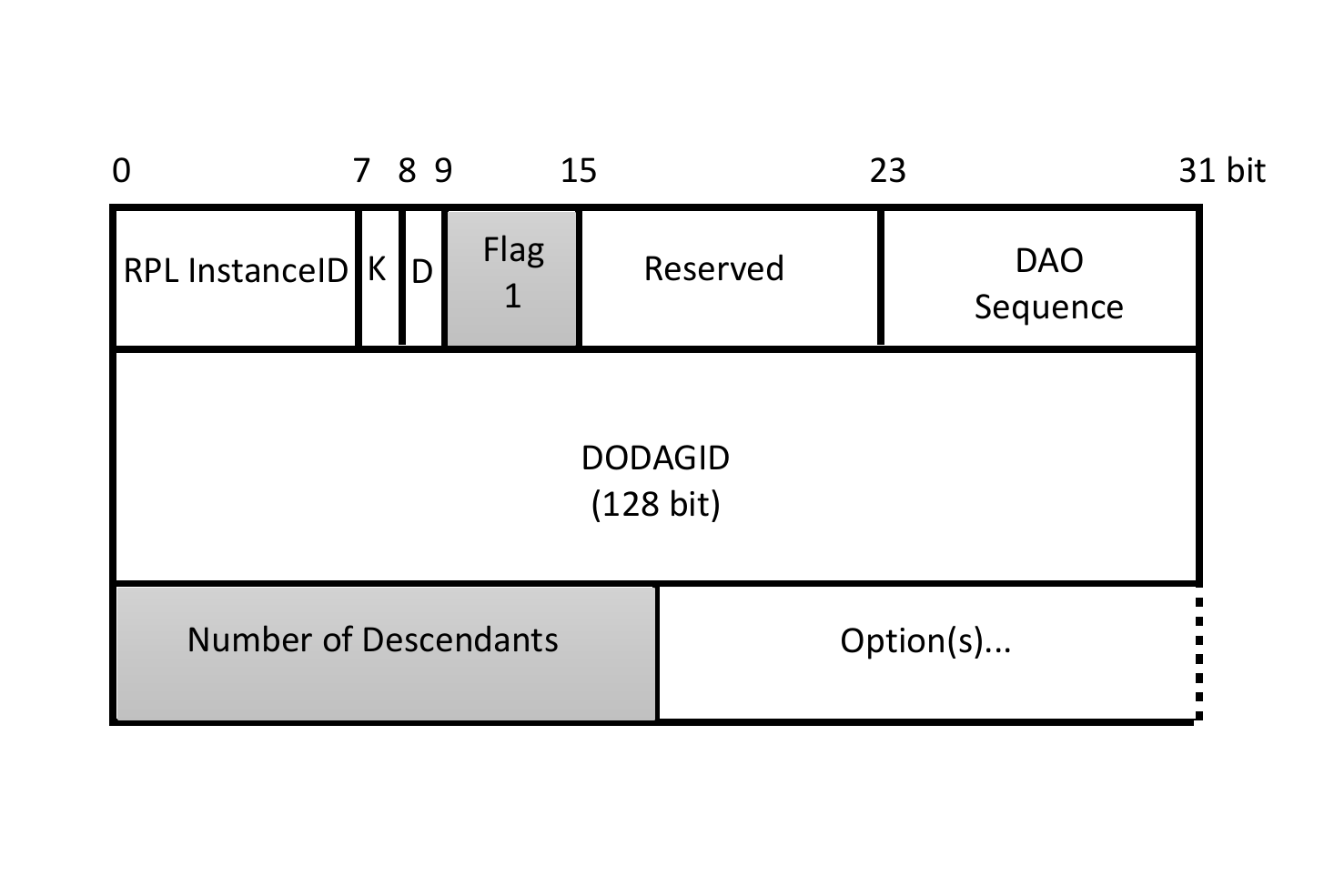}
	\caption{$DAO_{MHCL}$ message.}\label{fig:daomhcl}
	\end{figure}

Messages of type $DAO_{MHCL}$, as in RPL, are used in the upward routes, from
child to parent.
This message has two functions: in the greedy approach it informs
the preferred parent of a node that the latter is the former's child; in the
aggregate approach, it carries the number of a node's descendants, used in the
aggregation phase. We modified the $DAO_{RPL}$ message
by setting the \textit{flag} field to $1$ and sending the number of respective
descendants in the \textit{options} field (Figure~\ref{fig:diomhcl}).
$DAOACK_{MHCL}$ messages are used to acknowledge the reception of a $DAO_{MHCL}$
message. 

\subsection{Communication routines and timers}\label{subsec:timers}


The greedy approach MHCL-G is comprised of the following one-hop communication
routines:

\begin{enumerate}

  \item \textbf{Inform preferred parent:} If a node is not the root, it must
  inform its preferred parent that it is one of its children and requires an address
  (or range of addresses), sending a $DAO_{MHCL}$ message. A node must wait
  until the list of potential parents is stabilized and only then inform the
  chosen parent. MHCL uses Trickle-inspired timers to make such decisions (see
  Algorithm~\ref{alg:timerParent}). In
  Algorithm~\ref{alg:timerParent}, two constants are used: $DIO_{min}$ is a
  timer used by RPL for a node to set the minimum waiting time before sending a
  DIO message. $spChild$ is a \textit{stabilization parameter} used to decide
  when the preferred parent is stable (refer to
  Table~\ref{tab:stabilizationParams} for parameter values used in the simulations). Once the parent choice becomes
  stable, the local variable $parentDefined$ is set to TRUE and a
  $DAO_{MHCL}$ is send to the parent to inform about the decision (lines 12--18,
  Algorithm~\ref{alg:timerParent}).
  
	\item \textbf{Count children:}  Each non-leaf node must receive preferred
	parent messages ($DAO_{MHCL}$) from its children, save the list of
  children\footnote{Note that in standard implementations of acyclic structures in protocols such as RPL and CTP, a node does not save a list of
  children, but only a list of potential parents}, update the counter of
  children, and acknowledge the packet reception of each child (by sending a
  $DAOACK_{MHCL}$). The counting process goes on until the information becomes
  stable, which is controlled by parameter
  $spParent$. Once the number of
  children becomes stable, the local variable $childrenDefined$ is set to TRUE
  (line 13, Algorithm~\ref{alg:timerChildren}). 
  
	\item \textbf{Receive address range:} If a node is not the root, it
  must wait until it receives the address range allocated to it by its parent
  (in a $DIO_{MHCL}$ message) and acknowledge it by sending a $DIOACK_{MHCL}$
  message.
 
 \item \textbf{Send address range to children:} The address space partition and
  allocation is started by the root, once it determines the number of (direct)
  children ($childrenDefined$ is TRUE). The available address space is
  partitioned equally between the known children (keeping a reserve for future or delayed connections, as explained in Section~\ref{subsec:addr-space}). Once the
  address space is partitioned, the corresponding address range is
  sent in a $DIO_{MHCL}$ message to each child (see
  Algorithm~\ref{alg:addrDistr}). For non-root nodes, once the number of
  children is known and the address space has been received from the parent
  (in a $DIO_{MHCL}$ message), the address partitioning is done in the same way.
 
 \item \textbf{Delayed connections:} if a $DAO_{MHCL}$ message from a new child
  node is received after the address space had already been partitioned and
  assigned, then the address allocation procedure is repeated using the reserved
  address space.
\end{enumerate}

\begin{algorithm}[h]
\caption{MHCL: Preferred parent timer}\label{alg:timerParent}
\begin{algorithmic}[1]
\State parentDefined = FALSE;
\State maxTime = $spChild * DIO_{min}$ ;
\State timer = rand($1/2 * DIO_{min}$, $DIO_{min}$]; \Comment{reset timer}
\While{NOT parentDefined}
	\If {NOT-ROOT \textbf{and} TIMER-OFF} 
		\If {PARENT-CHANGED}
			\State reset timer;
		\Else	
		\If {$timer < maxTime$} 
			\State timer *= 2; \Comment{double timer}
		\Else
			\State parentDefined = TRUE; \Comment{MHCL-G-A}
			\If {MHCL-G} \Comment{greedy approach}
				\State send $DAO_{MHCL}$ to parent; 
				\If{NO $DAOACK_{MHCL}$} 
					\State send $DAO_{MHCL}$ to parent; \Comment{retry}
				\EndIf							
			\EndIf
		\EndIf
		\EndIf
	\EndIf	
\EndWhile
\end{algorithmic}
\end{algorithm}

\begin{algorithm}
\caption{MHCL-G: Children counter timer}\label{alg:timerChildren}
\begin{algorithmic}[1]
  \State childrenDefined = FALSE;	
  \State maxTime = $spParent *DIO_{min}$;
  \State timer = rand($1/2 * DIO_{min}$, $DIO_{min}$]; \Comment{reset timer}
  \State count = 0; \Comment{counts number of children by receiving
  $DAO_{MHCL}$ messages}
  \While{NOT childrenDefined}
  \If{TIMER-OFF}
		\If{COUNT-CHANGED}
			\State reset timer;
		\Else
			\If {$timer < maxTime$} 
				\State timer *= 2;
			\Else
				\State childrenDefined = TRUE;
			\EndIf
		\EndIf
  \EndIf
  \EndWhile
\end{algorithmic}
\end{algorithm}

\begin{algorithm}
\caption{MHCL: IPv6 address distribution}\label{alg:addrDistr}
\begin{algorithmic}[1]
  \State STABLE = FALSE;
  \If{MHCL-G} \Comment{Greedy MHCL-G}
  	\State STABLE = childrenDefined;
  \Else \Comment{Aggregate MHCL-A}
  	\State STABLE = descendantsDefined \textbf{or} NOT-ROOT;
  \EndIf
    		
  \If {STABLE \textbf{and} (IS-ROOT \textbf{or} RECEIVED-$DIO_{MHCL}$)} 
  	\State partition available address space;
	\For{each child $c_i$} 
		\State send $DIO_{MHCL}$ to $c_i$; \Comment{send IPv6 ``range''}
		\If{NO $DIOACK_{MHCL}$}
			\State send $DIO_{MHCL}$ to $c_i$; \Comment{retransmit}
		\EndIf
  	\EndFor		
  \EndIf
\end{algorithmic}
\end{algorithm}

The aggregate approach MHCL-A firstly executes an aggregation procedure to
compute the total number of descendants of each node. If a node is not the root,
and it has defined who the preferred parent is (parentDefined is TRUE) it starts
by sending a $DAO_{MHCL}$ message with $count=0$ (see
Algorithm~\ref{alg:bottomUpNotRoot}). Then it waits for $DAO_{MHCL}$ messages
from its children, updates the number of descendants of each child, and propagates the updated counter to the parent until its total number of
descendants is stable. If a node is the root, then it just updates the number of
descendants of each child by receiving  $DAO_{MHCL}$ messages until its total
number of descendants is stable (see Algorithm \ref{alg:bottomUpRoot}).
Parameters $spLeaf$ and $spRoot$ are used to define stabilization criteria in
non-root nodes and the root node, respectively.
Once the aggregation phase is completed, the root's local variable
$descendantsDefined$ is set to TRUE, and the address allocation process is
started by the root and propagated toward the leaves, as in the greedy approach (see
Algorithm~\ref{alg:addrDistr}).

\begin{algorithm}
\caption{MHCL-A: Aggregation
timer (non-root nodes)}\label{alg:bottomUpNotRoot} 
\begin{algorithmic}[1]
  \State maxTimeLeaf = $spLeaf *DIO_{min}$; 
  \State timer = rand($1/2 * DIO_{min}$, $DIO_{min}$]; \Comment{reset timer}
  \State count = 0; \Comment{counts total number of descendants by receiving
  $DAO_{MHCL}$ messages}
\While{NO-$DIO_{MHCL}$-FROM-PARENT} 
\State \Comment{while has not received IPv6 address range}
\If{NOT-ROOT \textbf{and} TIMER-OFF}
		\If{parentDefined \textbf{and} $(count<1)$}
			\State send $DAO_{MHCL}$ to parent; 
			\State \Comment{trigger aggregation}
		\EndIf
		\If {COUNT-CHANGED}
			\State send $DAO_{MHCL}$ to parent; 
			\State reset timer;		
		\Else	\If{$timer < maxTimeLeaf$} 
					\State timer *= 2;
				\EndIf
		\EndIf
\EndIf
\EndWhile
\end{algorithmic}
\end{algorithm}

\begin{algorithm}
\caption{MHCL-A: Aggregation timer (Root)}\label{alg:bottomUpRoot}
\begin{algorithmic}[1]
  \State descendantsDefined = FALSE;	
  \State maxTimeRoot = $spRoot *DIO_{min}$; 
  \State timer = rand($1/2 * DIO_{min}$, $DIO_{min}$]; \Comment{reset timer}
  \State count = 0; \Comment{counts total number of descendants by receiving
  $DAO_{MHCL}$ messages}
\While{NOT descendantsDefined}
\If{IS-ROOT \textbf{and} TIMER-OFF}
		\If{COUNT-CHANGED} 
			\State reset timer;
		\Else
			\If{$timer < maxTimeRoot$}
				\State $timer *= 2$;
			\Else
				\State descendantsDefined = TRUE;
			\EndIf
		\EndIf
\EndIf
\EndWhile
\end{algorithmic}
\end{algorithm}

\subsection{Routing tables and forwarding}\label{subsec:routing}

After the address allocation is complete, each (non-leaf) node stores a routing
table for downward traffic, with an entry for each child. Each table entry
contains the final address of the address range allocated to the corresponding
child, and all table entries are sorted in increasing order of the final address
of each range. In this way, message forwarding can be performed in linear
time using one comparison operation per table entry (see Algorithm~\ref{alg:forward}).
Of course, binary search could also be used, but considering that routing
table size is limited by the number of direct children of a node, we opted for
simplicity.

\begin{algorithm}
\caption{MHCL: Message forwarding}\label{alg:forward}
\begin{algorithmic}[1]
  \State $i = 0$;
	\While {$IPv6-dest > routingTable[i].finalRangeAddr$}
		\State $i++$;
	\EndWhile
	\State forward message to $routingTable[i].childAddr$;
\end{algorithmic}
\end{algorithm}

\section{Complexity Analysis}\label{sec:analysis}

We now turn our attention to the time ($Time(MHCL)$) and message
($Message(MHCL)$) complexity of MHCL, assuming a synchronous communication model
with point-to-point message passing, i.e., that all nodes start executing the algorithm simultaneously and that time is divided into synchronous rounds, such that, whena message is sent from node $v$ to its neighbor $u$ at time-slot $t$, it must arrive at $u$ before time-slot $t+1$.

Note that MHCL requires that an underlying acyclic topology (say tree $T$) has
been constructed by the network before the address allocation starts, i.e., every node knows who its preferred parent is ($parentDefined==TRUE$).
The greedy approach MHCL-G has two  phases: (1) a one-hop communication routine, run
in parallel at all nodes, in which each node informs its preferred parent about
its decision, the parent sends an acknowledgment, and, by the end of this
routine, every node knows how many children it has; (2)
a broadcast from the root to all nodes in the tree $T$ of address allocation
information, each message being acknowledged by the recipient of the address
range. The time complexity of routine (1) is $2$, and the message complexity is
$2(n-1)$. The time complexity of routine (2) is $depth(T)$ and message complexity is $2(n-1)$.

The aggregate approach MHCL-A also has two phases: (1) number of descendants
aggregation, using a convergecast routine, in which leaf nodes start sending
messages to their parents and parents aggregate the received information and
forward it upwards until the root is reached; (2) the root broadcasts address
allocation information downward along the tree $T$, each message being
acknowledged. The time complexity of
each phase is $depth(T)$ and the message complexity is $2(n-1)$.

The overall complexity of MHCL in the synchronous point-to-point message passing
model is summarized in the following theorem.

\begin{theorem} For any network of size $n$ with a spanning tree $T$ rooted at
node $root$, $Message(MHCL(T, root))= O(n)$ and $Time(MHCL(T, root)) =
O(depth(T))$. This message and time complexity is asymptotically optimal.
\end{theorem}
\begin{proof}
MHCL is comprised of a tree broadcast and a tree convergecast (in the aggregate
approach). In the broadcast operation, a message (with address allocation
information) must be sent to every node by the respective parent, which needs
$\Omega(n)$ messages.
Moreover the message sent by the root must reach every node at
distance $depth(T)$ hops away, which needs $\Omega(depth(T))$
time-slots. Similarly, in the convergecast operation, every node must send a
message to its parent after having received a message from its children, which
needs $\Omega(n)$ messages. Also, a message sent by every leaf node must reach
the root, at distance $\leq depth(T)$, which needs $\Omega(depth(T))$
time-slots.
\end{proof}

Note that, in reality, the assumptions of synchrony and point-to-point message
delivery do not hold in a 6LoWPAN. The moment in which each node joins the tree,
i.e., sets $parentDefined=TRUE$ varies from node to node, such that nodes closer
to the root tend to start executing the address allocation protocol earlier than
nodes farther away from the root. Moreover, collisions and node and
link failures can cause delays and prevent messages from being delivered. We
analyze the performance of MHCL in an asynchronous model with collisions and
transient node and link failures through simulations in Section~\ref{sec:results}.

\section{Experimental Results}
\label{sec:results}

In this section we evaluate the performance of MHCL through simulations.
In Section \ref{subsec:settings} we describe the simulations setup and
configuration parameters used.
In Section \ref{subsec:setup} we evaluate the setup phase of the network
in terms of time, number of control messages and address allocation success
rate. In Section \ref{subsec:appLayer} we evaluate message delivery
success rate at the application layer for top-down data flows.

\subsection{Simulation Settings}\label{subsec:settings}

MHCL was implemented as a subroutine of the RPL protocol in Contiki
OS \cite{Dunkels:2004}, replacing the standard static IPv6 address allocation
mechanism, in which the host address of a node is derived from its MAC address.
Our experimental results compare the performance of RPL using MHCL-G (greedy
approach) and MHCL-A (aggregate approach) against the static address allocation
(referred to as RPL in the plots). To perform the experiments, we used Cooja~\cite{Eriksson:2009},
a 6LoWPAN simulator provided by Contiki, with nodes
of type \textit{SkyMote}, with 10KB of RAM and 50 meter transmission range.

We simulated two types of topology: regular (with $n$ nodes distributed in a
regular grid with 35m pairwise distance) and uniform (with nodes distributed
uniformly at random over an area of side $(\sqrt{n}-2) * 35$m), with $n \in \{9, 25,
49, 81, 121, 169\}$. We also simulated two types of (transient) failures at the
packet level:
node (TX) and link (RX) failures.
In a failure of type TX, in $1-TX\%$ of transmissions, no device receives the
packet. In failures of type RX, in $1-RX\%$ of transmissions, only the
destination node does not receive the packet. We simulated TX and RX failures
separately, using failure rates of $0\%$, $5\%$, and $10\%$.

At the application layer, we used as reference an application code provided
within Cooja (rpl-UDP), which uses UDP in the transport layer, i.e. does not
implement packet acknowledgment or retransmission, and RPL in the network layer. We made some
changes in the application, so that each node sends a message to the root and
the root answers each of these messages. Thus, the application sends $n-1$ bottom-up messages and $n-1$
top-down messages. Nodes begin sending application messages after 180 seconds of
simulation. The simulation of each algorithm terminates when the last
application message has been answered, or when the time for this event to be
completed ends, in case of failures.

Table~\ref {tab:stabilizationParams} shows the parameters of the algorithms
presented in Section~\ref{sec:mhcl}, which are used to estimate the
stabilization time of the network as follows: each timer used in MHCL protocol
starts with a minimum value $I = 2^{DIO_{min}}$ ms (the same
parameter used in RPL) and a maximum value $sp * 2^{DIO_{min}}$ ms, where $sp$
is the stabilization parameter of each routine. 
(see the algorithms presented in Section ~\ref{sec:mhcl}).

\subsection{Setup Phase}\label{subsec:setup}

To evaluate MHCL algorithm, the setup time was defined as the time required for all nodes
to be addressed, excluding the not addressed nodes due to collisions and link or
node failures.
In the standard RPL protocol, we measured the time needed for the  root to save $n-1$ routes in the downward routing
table, i.e., the time needed for the root  to know the path to every destination in the
network. Figures \ref{fig:setup}, \ref{fig:setup_f}  and \ref{fig:setup_f_u} show the setup
time of RPL, MHCL-G and MHCL-A in scenarios without  failures, with failures in the regular 
topology and with failures in the uniform topology,  respectively. Note that we are dealing with
asynchronous algorithms and the start of each  node is randomized, which can affect the setup time.
Nevertheless, we can see that the setup  time of all protocols grows linearly with the height of the
DAG structure, matching the theoretical time complexity shown in
Section~\ref{sec:analysis}. If we consider a $95\%$ confidence interval, we cannot claim that one algorithm is faster than another in the regular
topology, because of overlapping confidence intervals (except maybe for height 24 in Figure~\ref{fig:setup},
where we can see that MHCL-G is faster than MHCL-A and RPL, with no overlapping).

We performed the t-test~\cite{Jain1991} for the setup time in the regular distribution without
failures, which revealed that, with $95\%$ of confidence, the setup times of the three algorithms
are statistically the same, as all three intervals contain zero. However, for the uniform topology,
there is a significant difference between MHCL and RPL. With $95\%$ confidence, we can say that in 
a network with and without failures, RPL presents a faster setup, since there is no overlap between confidence intervals.

\begin{figure}[t]
\centering
\includegraphics[width=0.6\columnwidth]{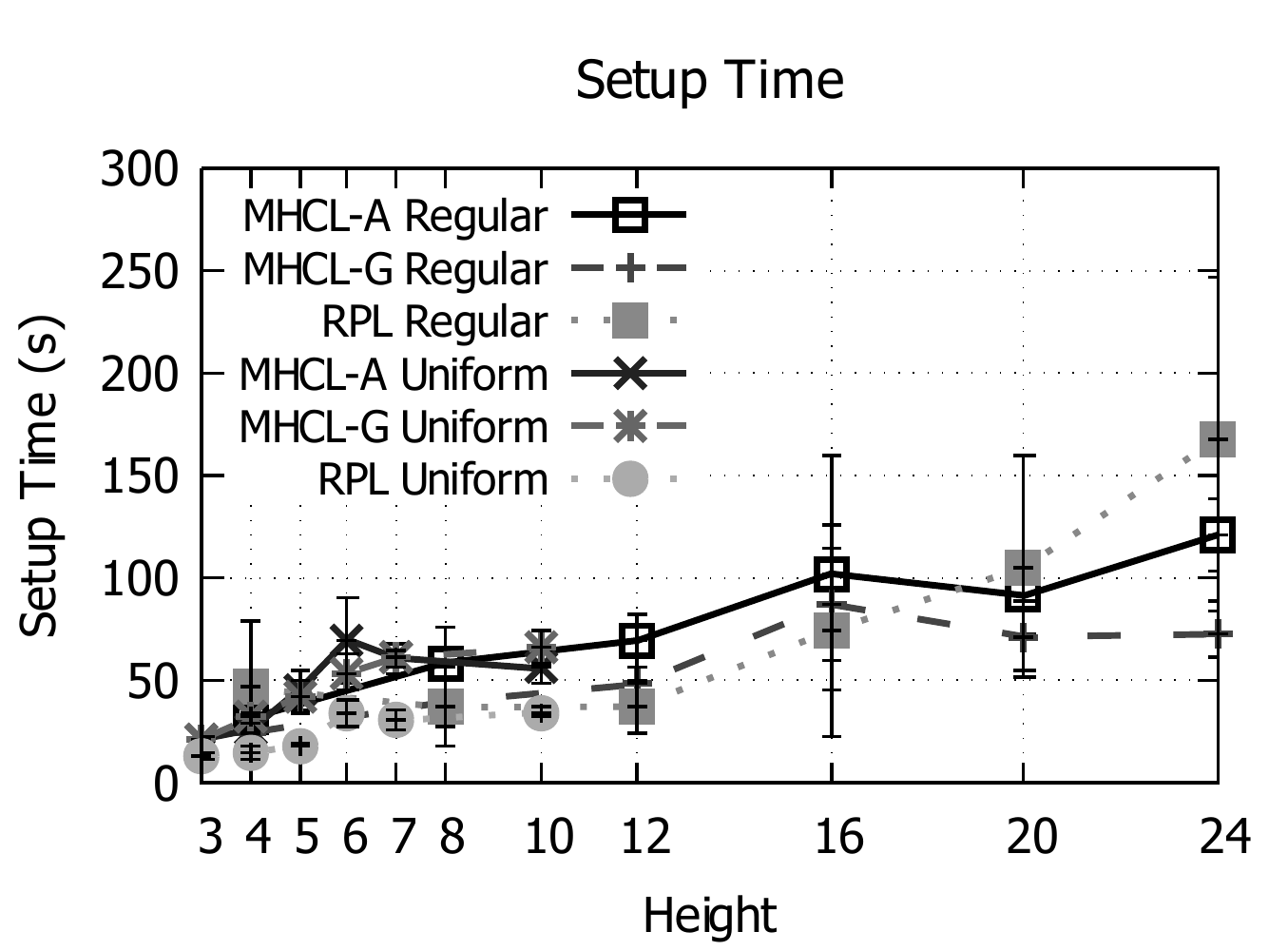}
\caption{Network's setup time.}\label{fig:setup}
\end{figure}

\begin{figure}[h]
\begin{center}
  \subfigure[TX Failure]
  {\includegraphics[width=.45\columnwidth]{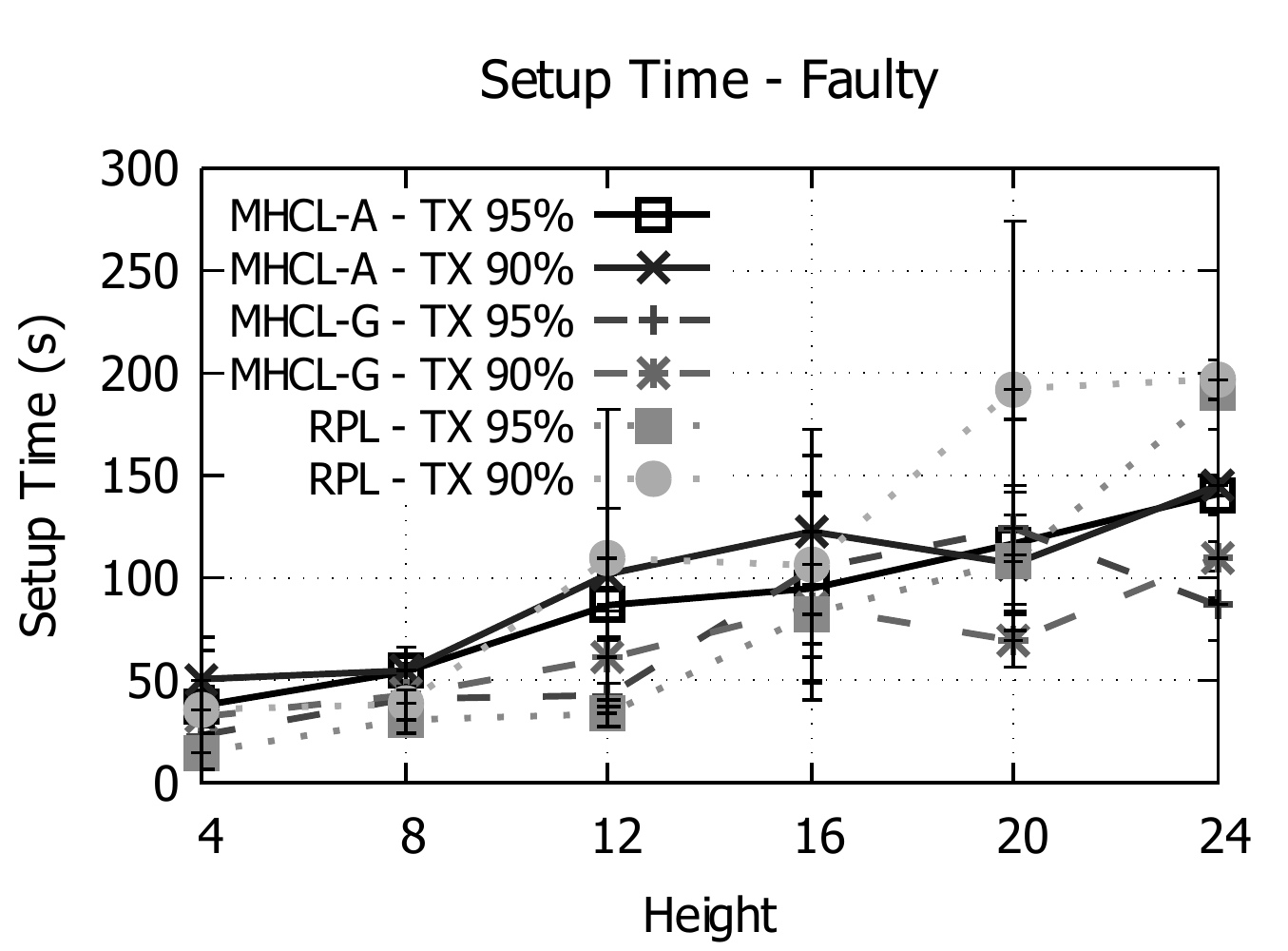}}
  \subfigure[RX Failure]
            {\includegraphics[width=.45\columnwidth]{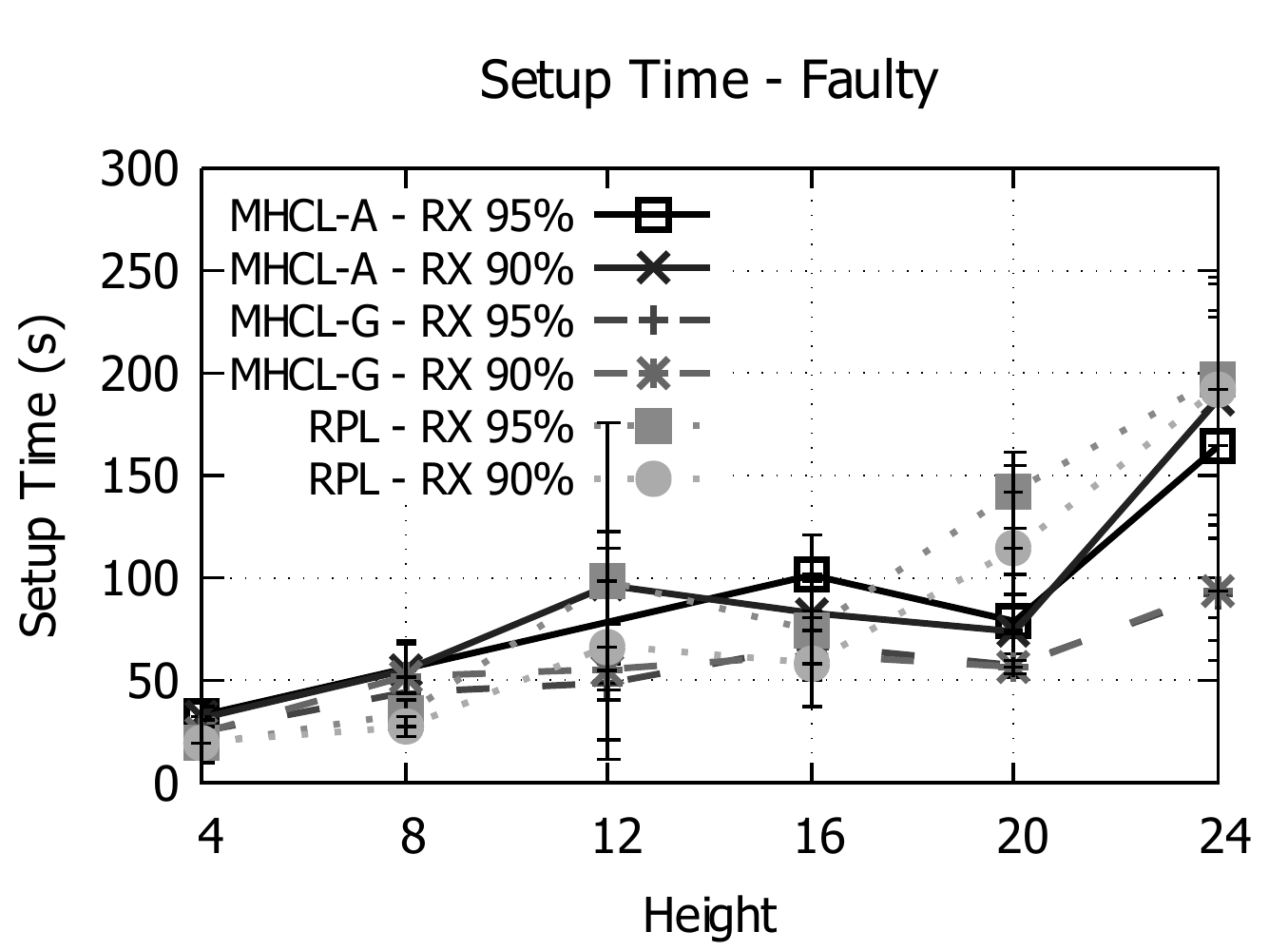}}
            \caption{Network's setup time - Faulty - Regular Topology.}\label{fig:setup_f}
\end{center}
\end{figure}

\begin{figure}[h]
\begin{center}
  \subfigure[TX Failure]
  {\includegraphics[width=.45\columnwidth]{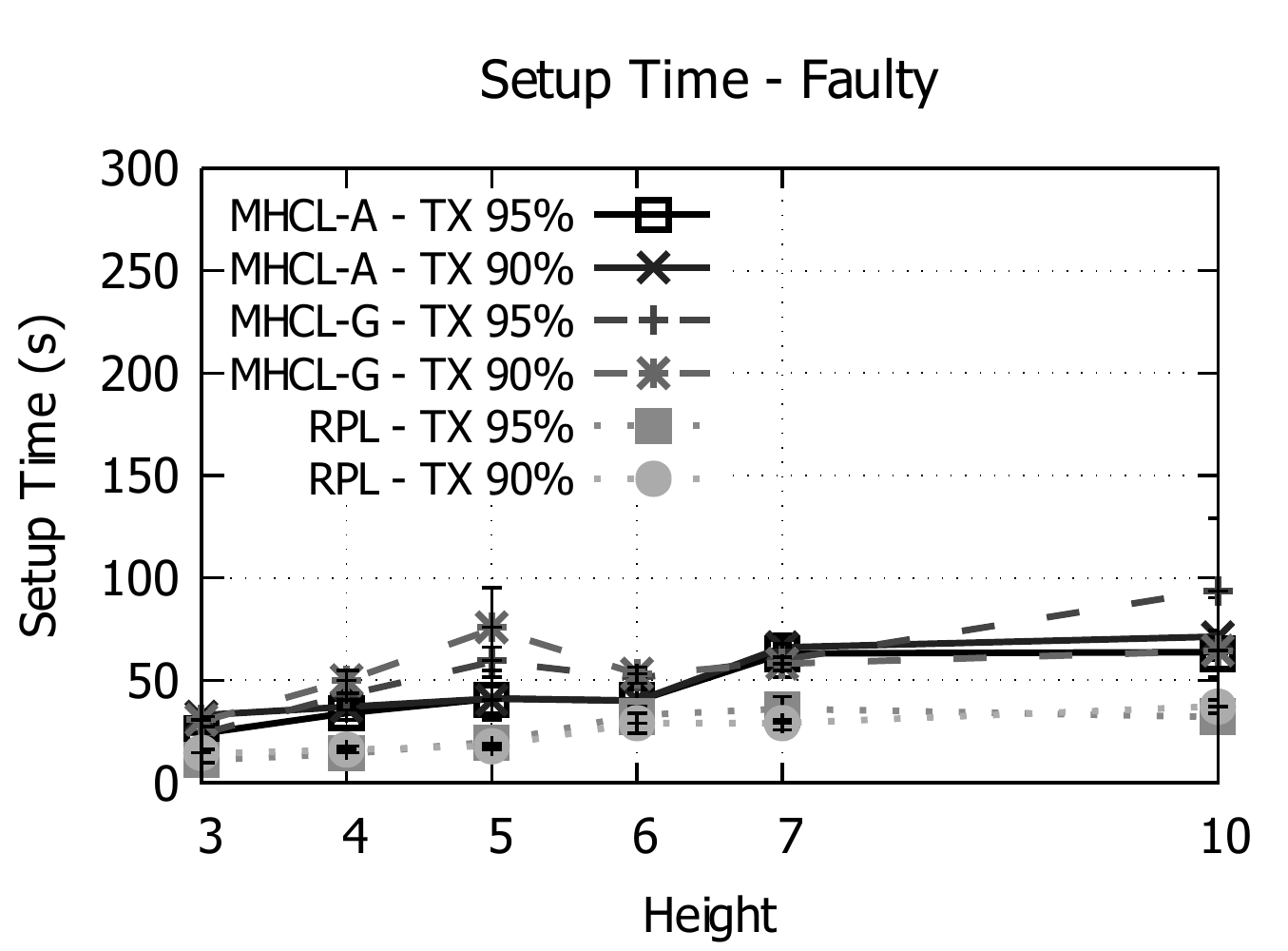}}
  \subfigure[RX Failure]
            {\includegraphics[width=.45\columnwidth]{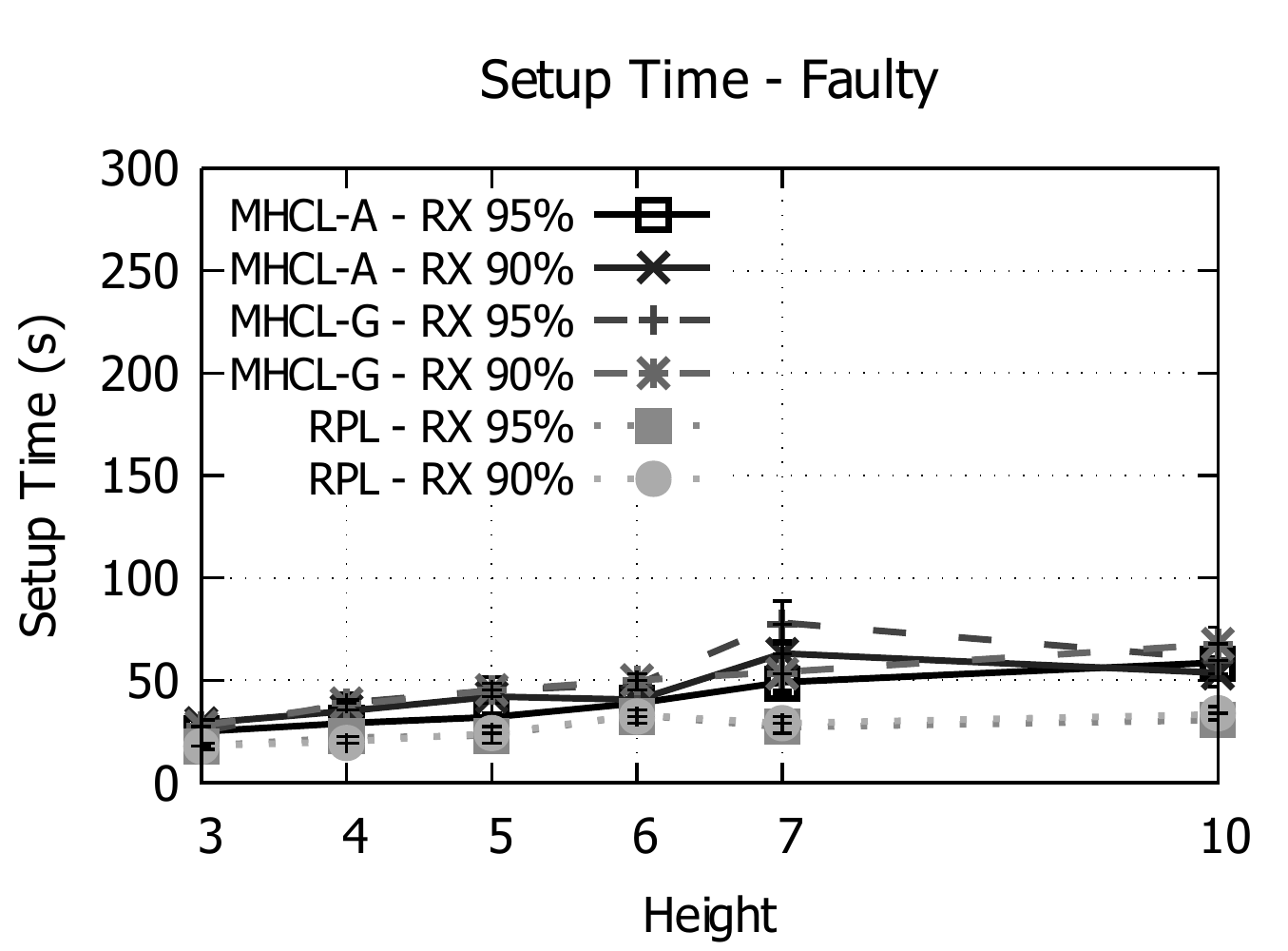}}
            \caption{Network's setup time - Faulty - Uniform Topology.}\label{fig:setup_f_u}
\end{center}
\end{figure}

To count the number of messages sent by each protocol, an interval of 180 seconds was defined. 
During this interval, the numbers of DAO and DIO messages were counted separately  (with the respective
acknowledgments). Figures \ref{fig:dio}, \ref{fig:dio_f} and \ref{fig:dio_f_u} show the number of DIOs and Figures
\ref{fig:dao}, \ref{fig:dao_f} and \ref{fig:dao_f_u} show the number of DAOs sent by RPL, MHCL-G and MHCL-A in scenarios
without and with failures in the regular and uniform topologies. It can be seen that,
despite collisions and transient failures of nodes and links, the number of both types of messages grows
linearly with the number of nodes, matching the theoretical message complexity
shown in Section~\ref{sec:analysis}.

The number of DIOs starts to differ statistically between algorithms,
with a confidence level of $95 \%$, after 121 nodes, for regular topology.
We expected no difference between the algorithms, since the mechanism for discovery and maintenance of DAG is not
different from RPL to MHCL. MHCL uses DIO messages for the address allocation mechanism, but the number of additional
messages is small compared to the total amount of DIO messages transmitted in the network. However, RPL starts sending
more DIO messages than MHCL as the network grows.

Since the route discovery mechanism, which uses DAO messages, is different in MHCL and RPL, the number of sent
messages of this type has a significant difference. With $95 \%$ confidence we
can say that MHCL sends less DAO messages than RPL in all simulated
scenarios with and without failures. In MHCL, after the
topology information has been collected, no more DAO messages are sent. In RPL,
route advertisement does not halt, which makes DAO messages
to be continuously transmitted.

\begin{figure}[t]
\centering
\includegraphics[width=0.6\columnwidth]{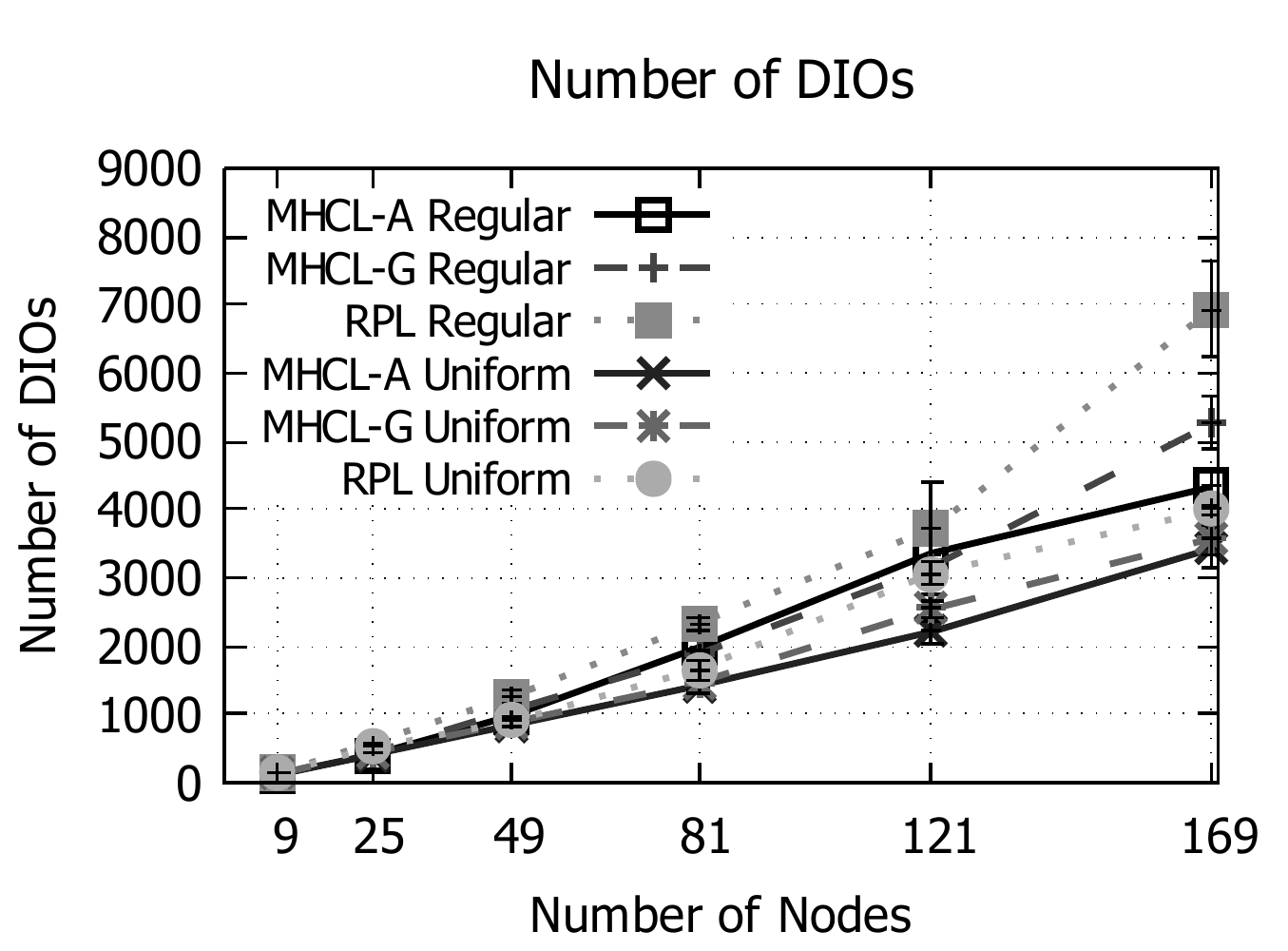}
\caption{Number of DIO messages.}\label{fig:dio}
\end{figure}

\begin{figure}[h]
\begin{center}
  \subfigure[TX Failure]
  {\includegraphics[width=.45\columnwidth]{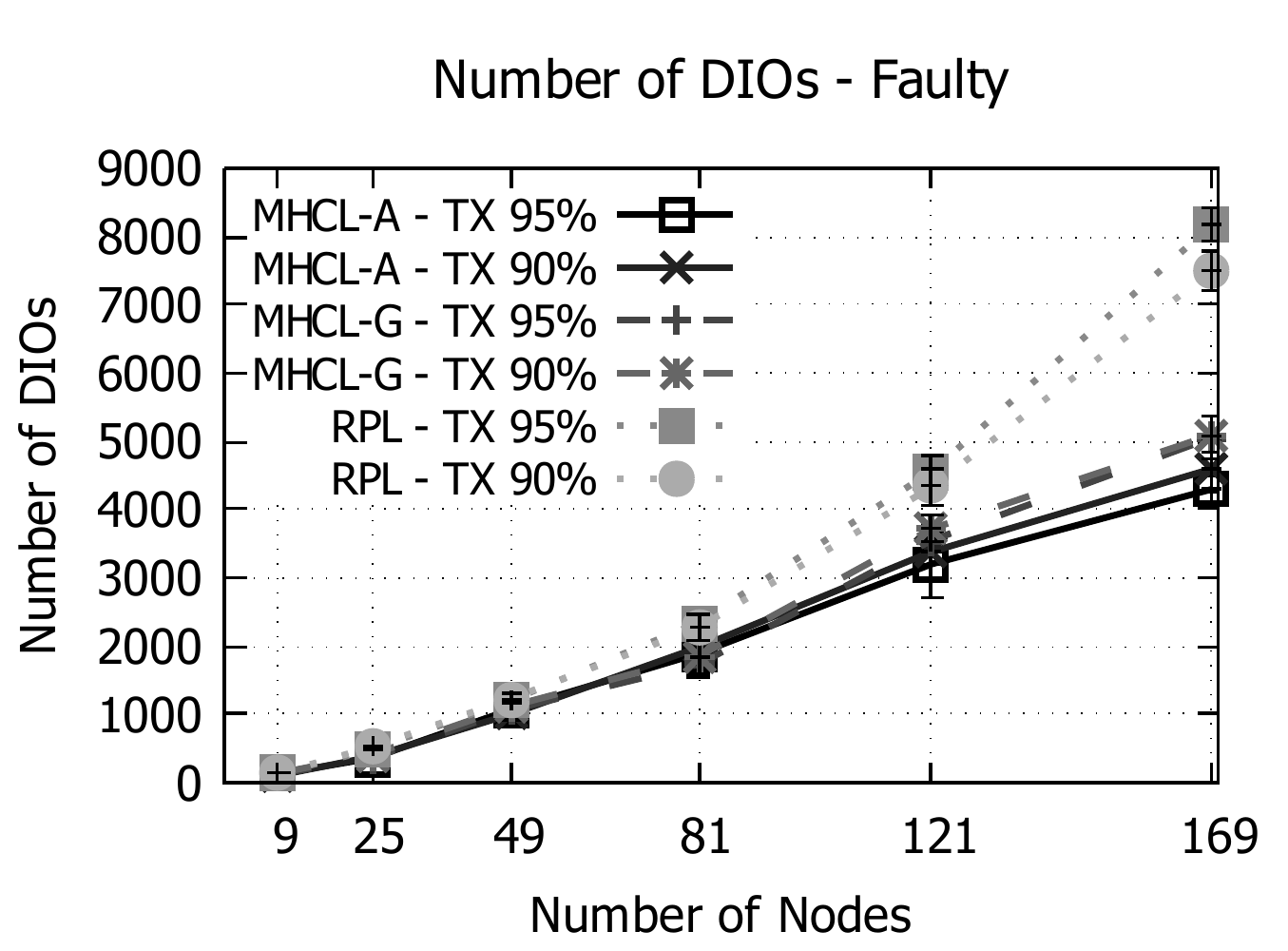}}
  \subfigure[RX Failure]
            {\includegraphics[width=.45\columnwidth]{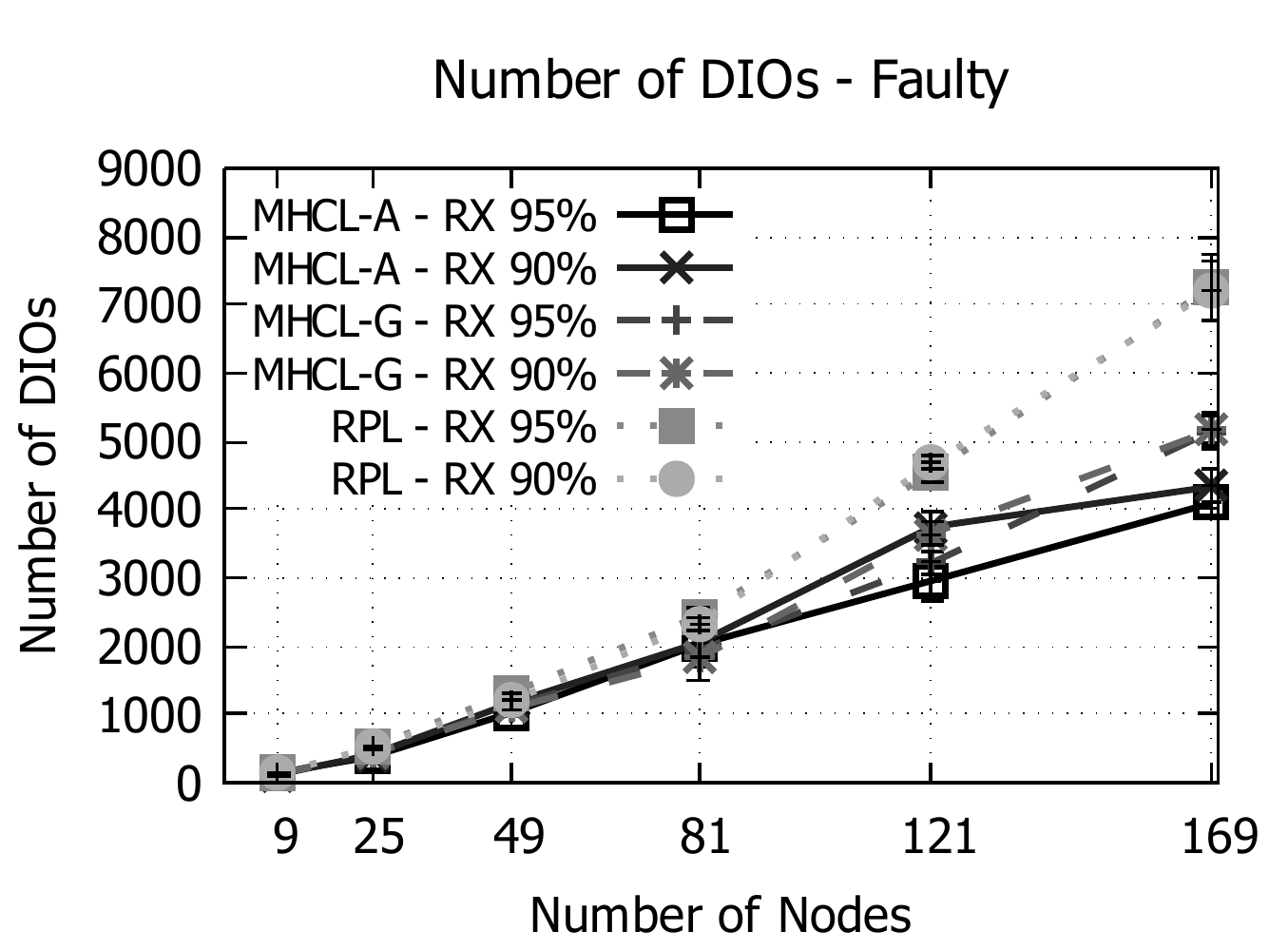}}
            \caption{Number of DIO messages - Faulty - Regular Topology.}\label{fig:dio_f}
\end{center}
\end{figure}

\begin{figure}[h]
\begin{center}
  \subfigure[TX Failure]
  {\includegraphics[width=.45\columnwidth]{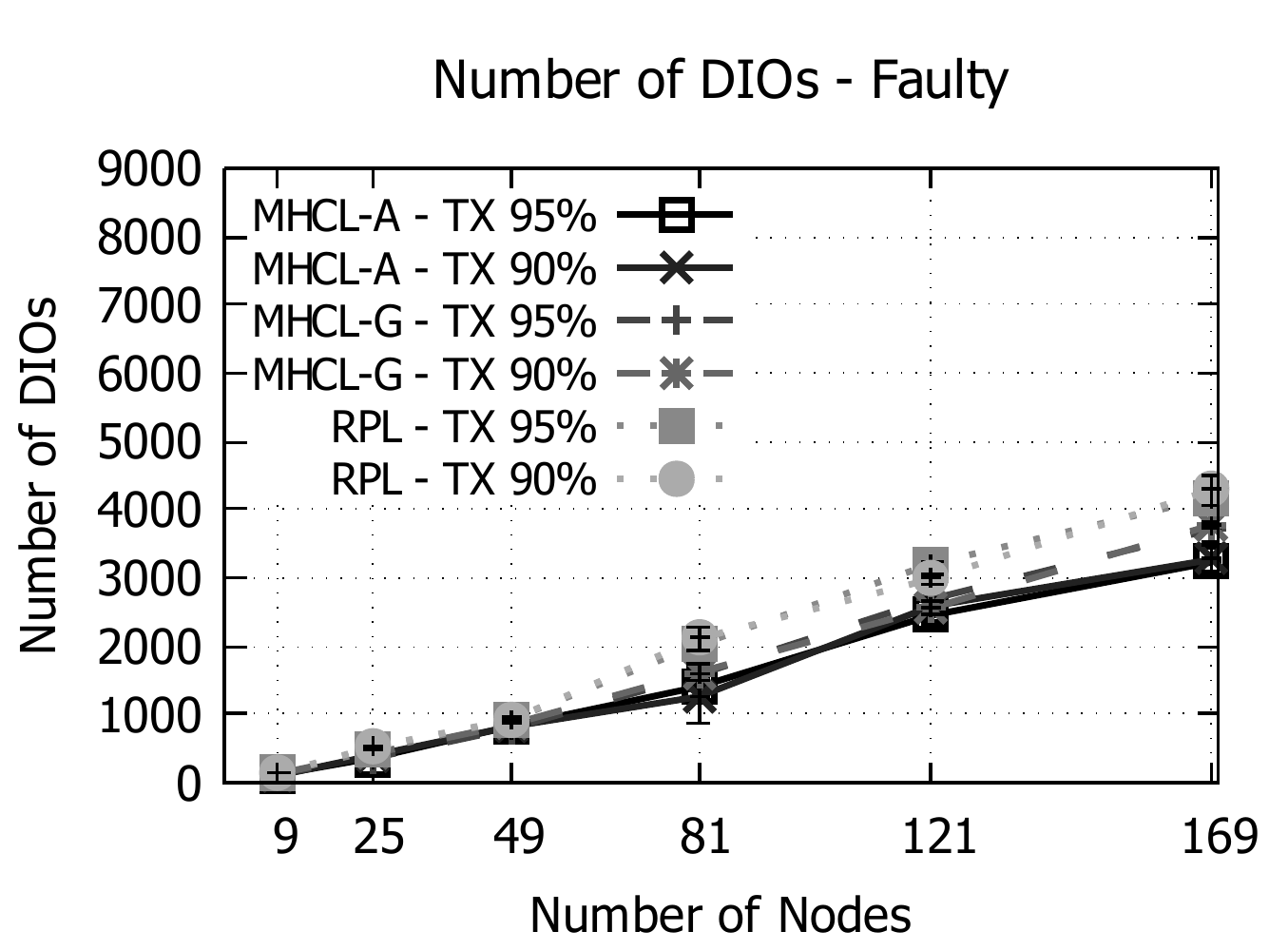}}
  \subfigure[RX Failure]
            {\includegraphics[width=.45\columnwidth]{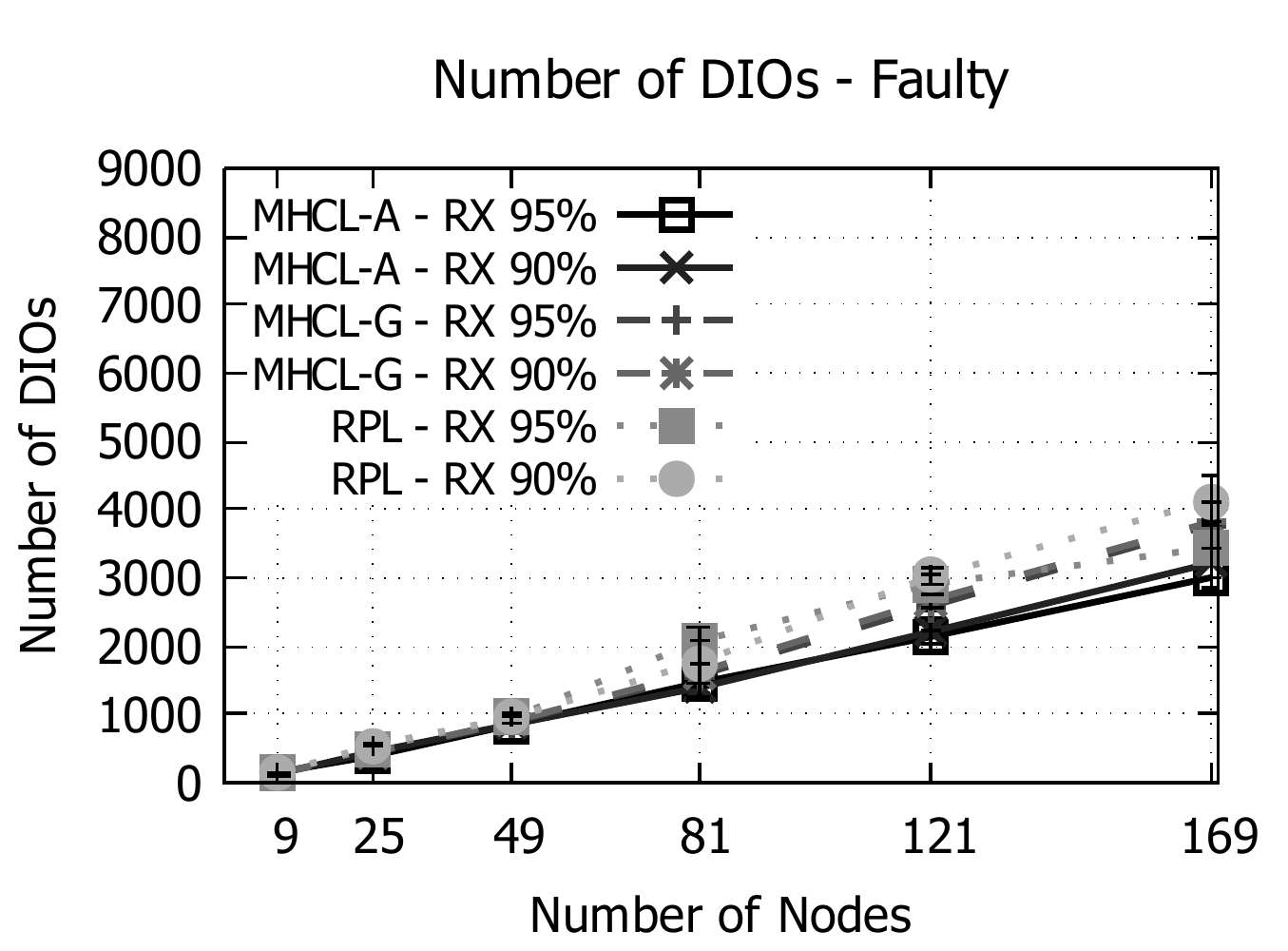}}
            \caption{Number of DIO messages - Faulty - Uniform Topology.}\label{fig:dio_f_u}
\end{center}
\end{figure}

\begin{figure}[t]
\centering
\includegraphics[width=0.6\columnwidth]{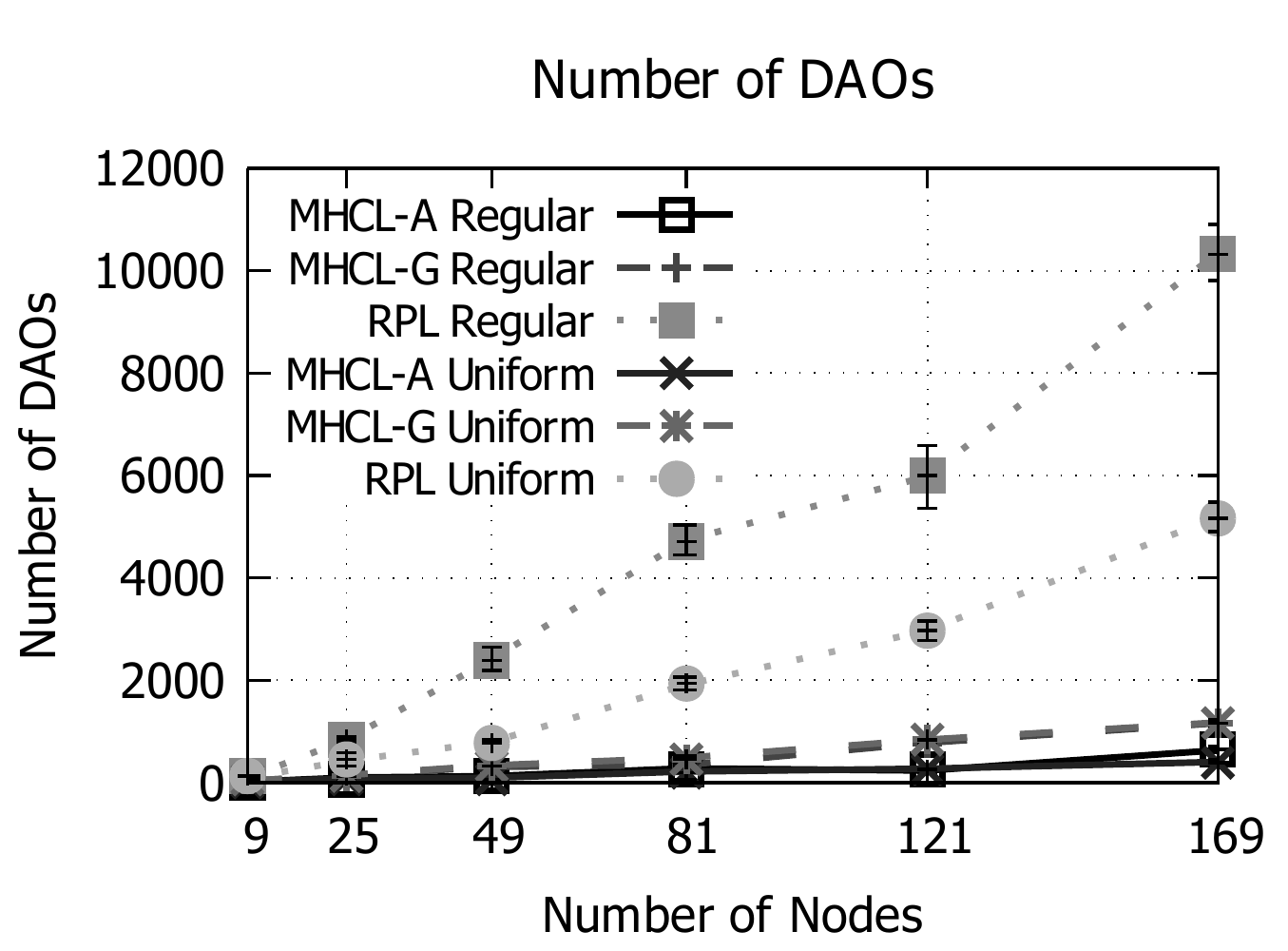}
\caption{Number of DAO messages.}\label{fig:dao}
\end{figure}

\begin{figure}[h]
\begin{center}
  \subfigure[TX Failure]
  {\includegraphics[width=.45\columnwidth]{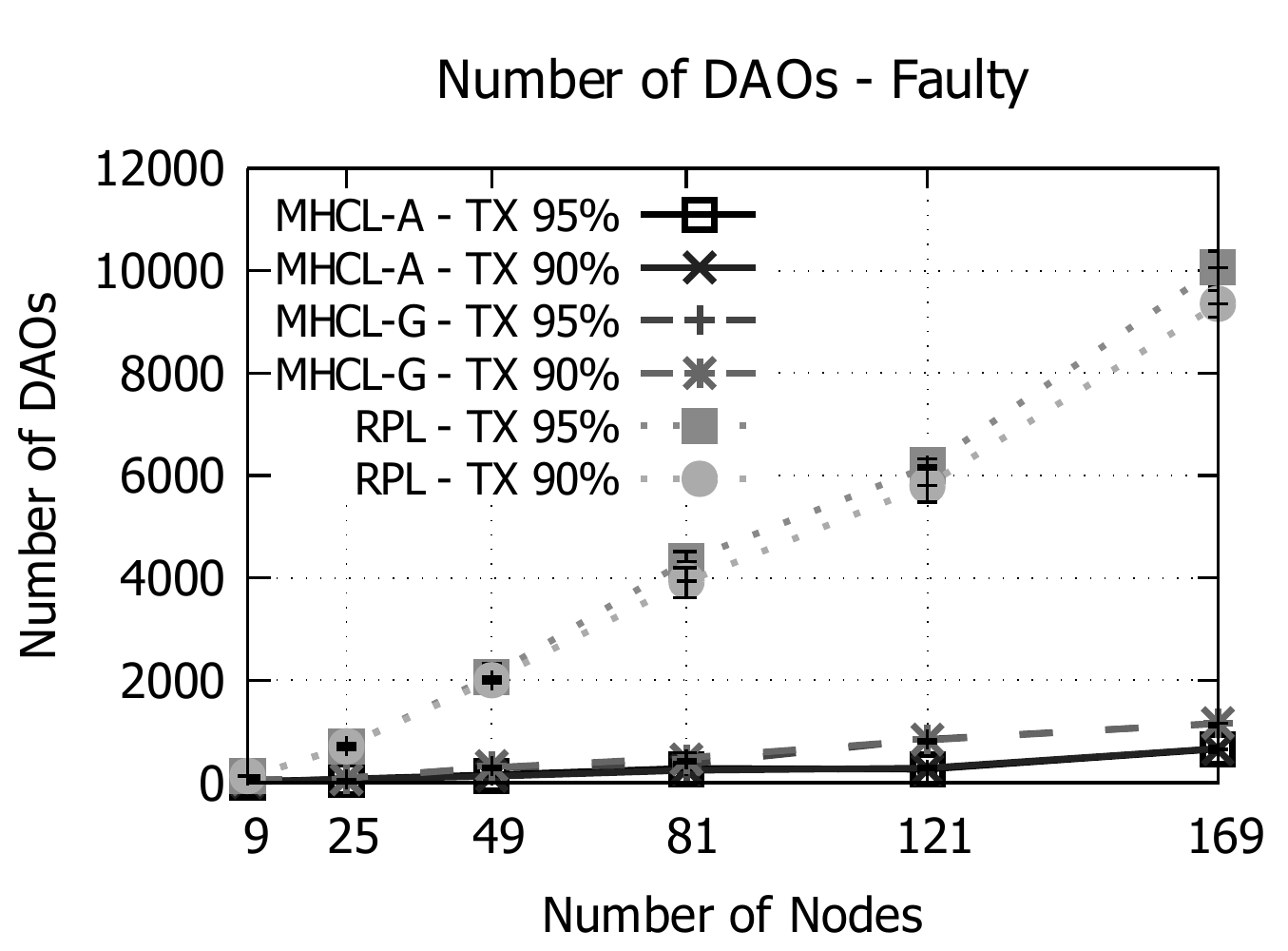}}
  \subfigure[RX Failure]
            {\includegraphics[width=.45\columnwidth]{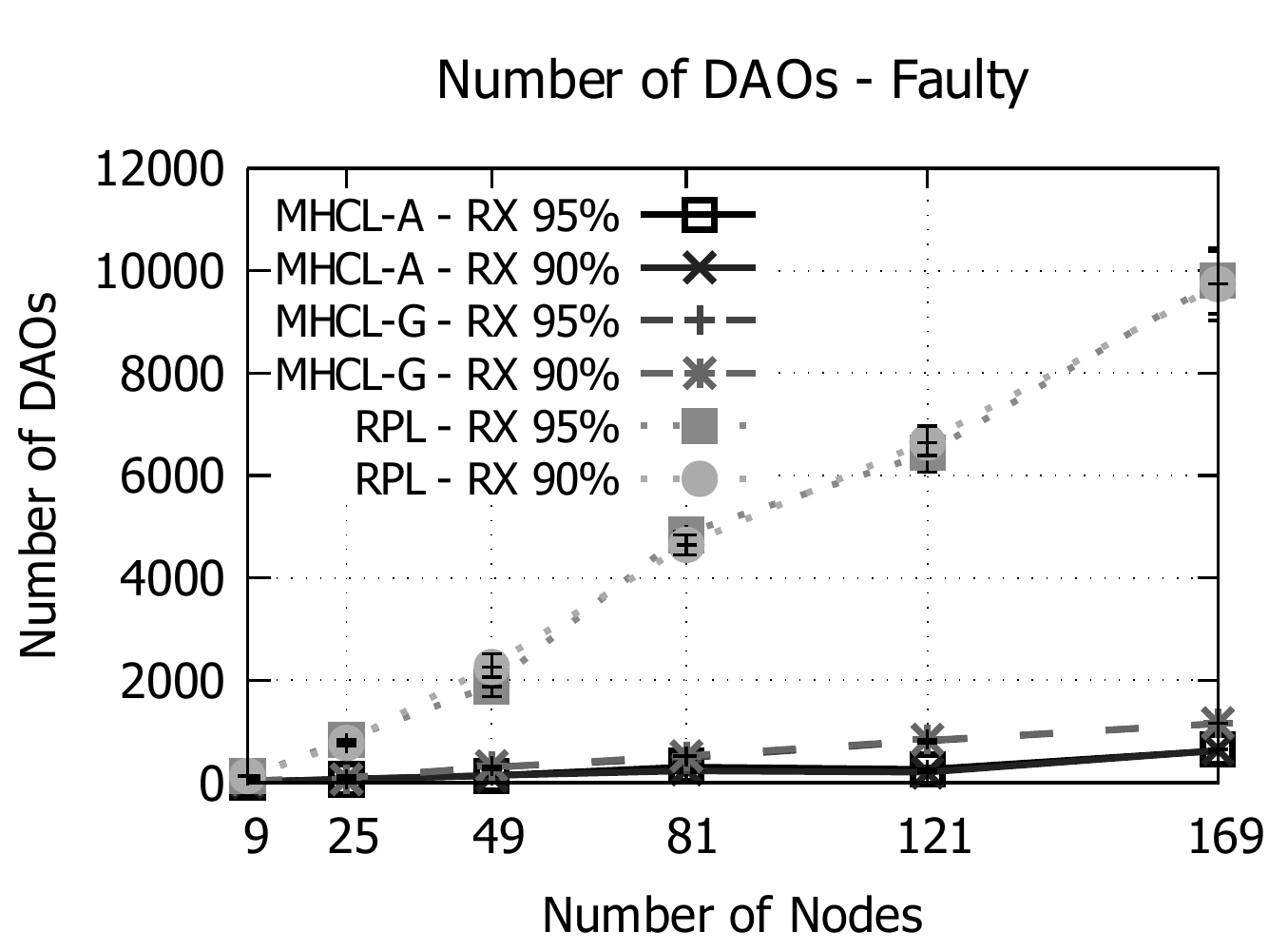}}
            \caption{Number of DAO messages - Faulty - Regular Topology.}\label{fig:dao_f}
\end{center}
\end{figure}

\begin{figure}[h]
\begin{center}
  \subfigure[TX Failure]
  {\includegraphics[width=.45\columnwidth]{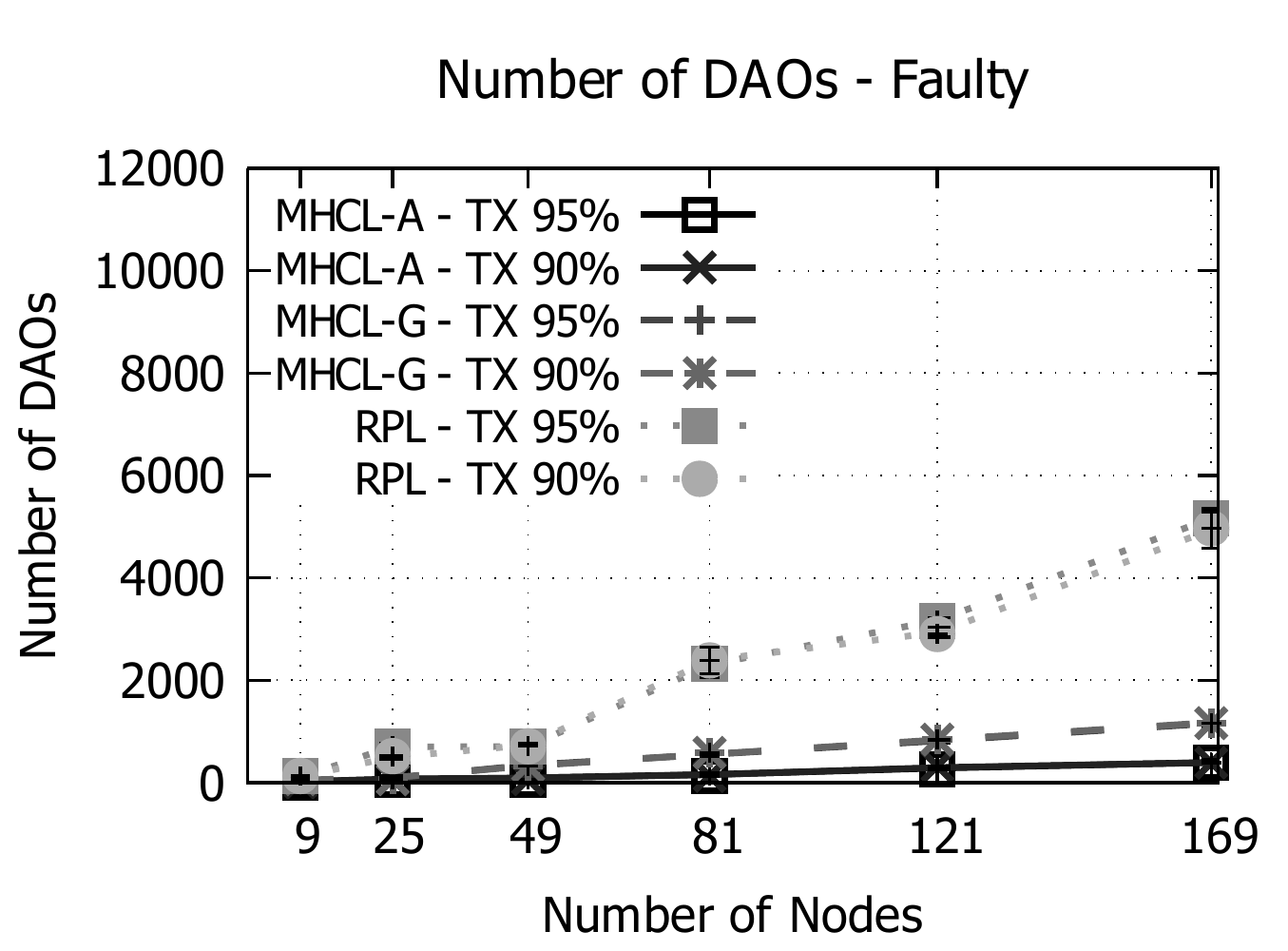}}
  \subfigure[RX Failure]
            {\includegraphics[width=.45\columnwidth]{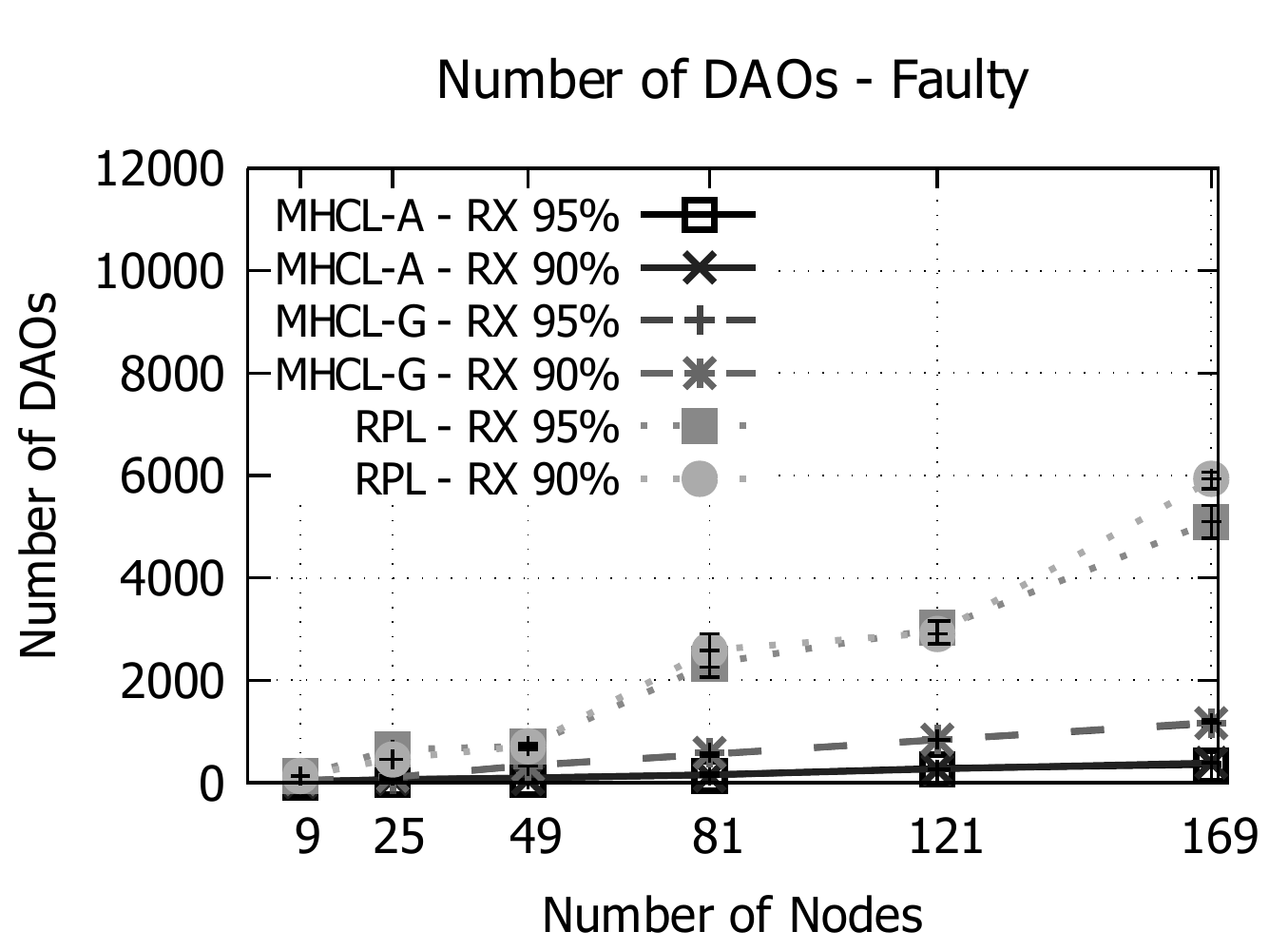}}
            \caption{Number of DAO messages - Faulty - Uniform Topology.}\label{fig:dao_f_u}
\end{center}
\end{figure}

MHCL performs IPv6 address allocation during the network setup phase. Since the network may have
collisions and node and link failures, some $DIO_{MHCL}$ messages can be lost, even using acknowledgments
($DIOACK_{MHCL}$) and retransmissions (in our simulations, up to 3
retransmissions were done). Figures \ref{fig:txend}, \ref{fig:txend_f} and
\ref{fig:txend_f_u} compare the addressing success rate of MHCL-A and MHCL-G in
scenarios without failures and with failures in the regular and uniform topologies.

We can observe that,
despite collisions, the addressing rate is between 70\% and 100\% in the non-faulty
scenario. When intermittent failures of nodes and links are present, however,
that average number decreased slightly, as expected.
With $95 \%$ confidence we can say that MHCL-G and MHCL-A have the same addressing rate,
since there is an overlap between the confidence intervals, in both regular and
uniform topologies.

\begin{figure}[t]
\centering
\includegraphics[width=0.6\columnwidth]{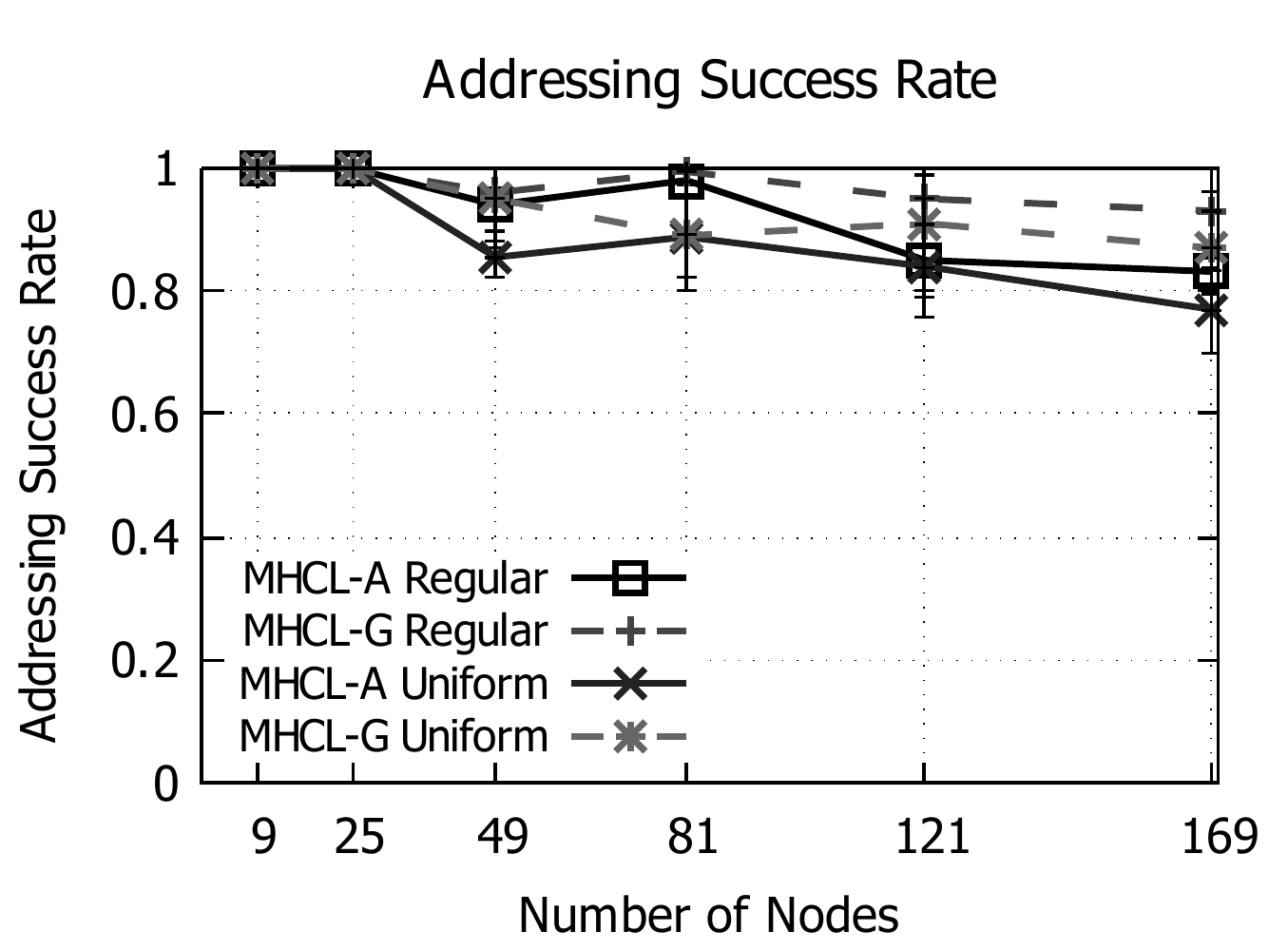}
\caption{Addressing rate.}\label{fig:txend}
\end{figure}

\begin{figure}[h]
\begin{center}
  \subfigure[TX Failure]
  {\includegraphics[width=.45\columnwidth]{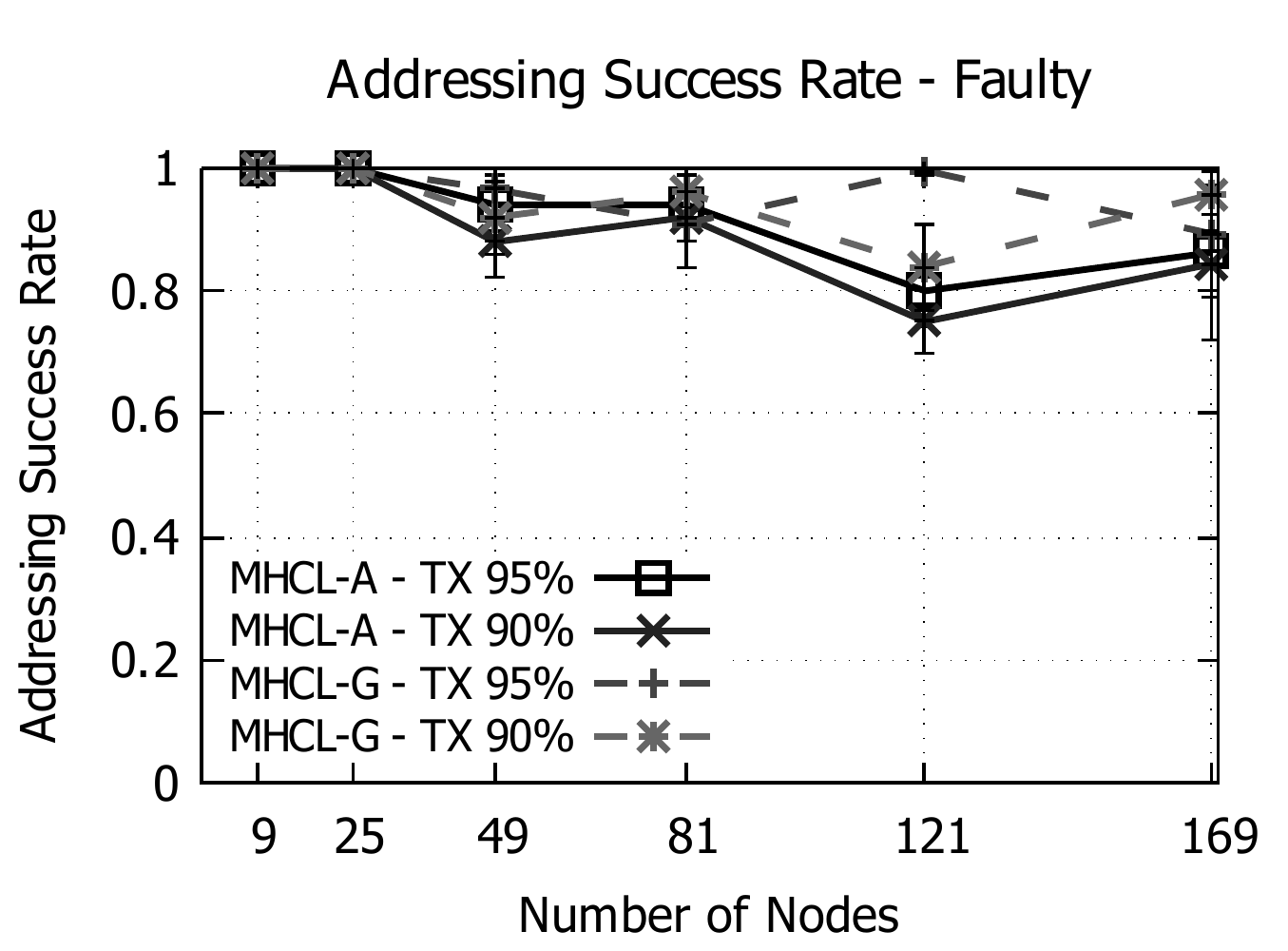}}
  \subfigure[RX Failure]
            {\includegraphics[width=.45\columnwidth]{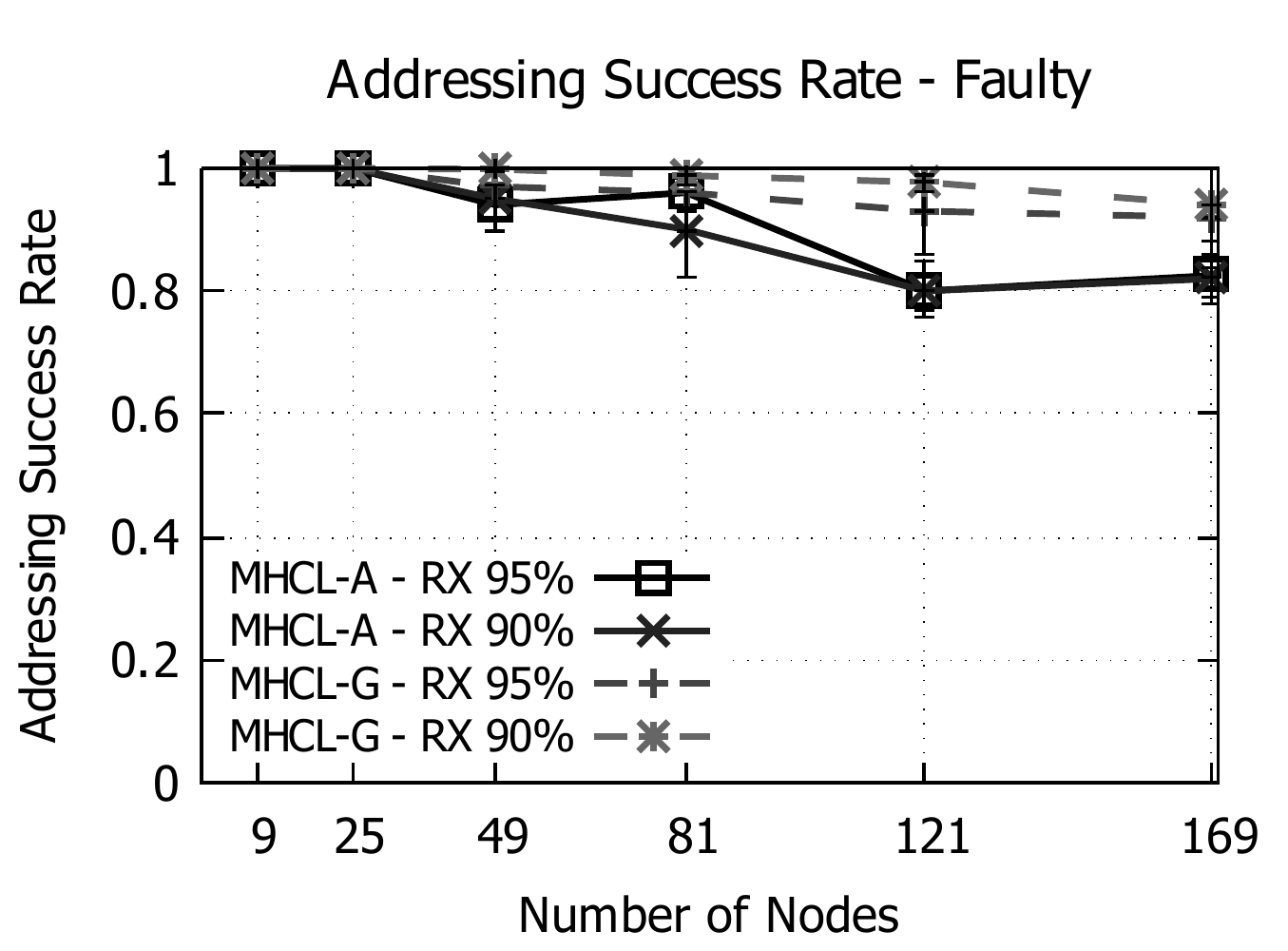}}
            \caption{Addressing rate - Faulty - Regular Topology.}\label{fig:txend_f}
\end{center}
\end{figure}

\begin{figure}[h]
\begin{center}
  \subfigure[TX Failure]
  {\includegraphics[width=.45\columnwidth]{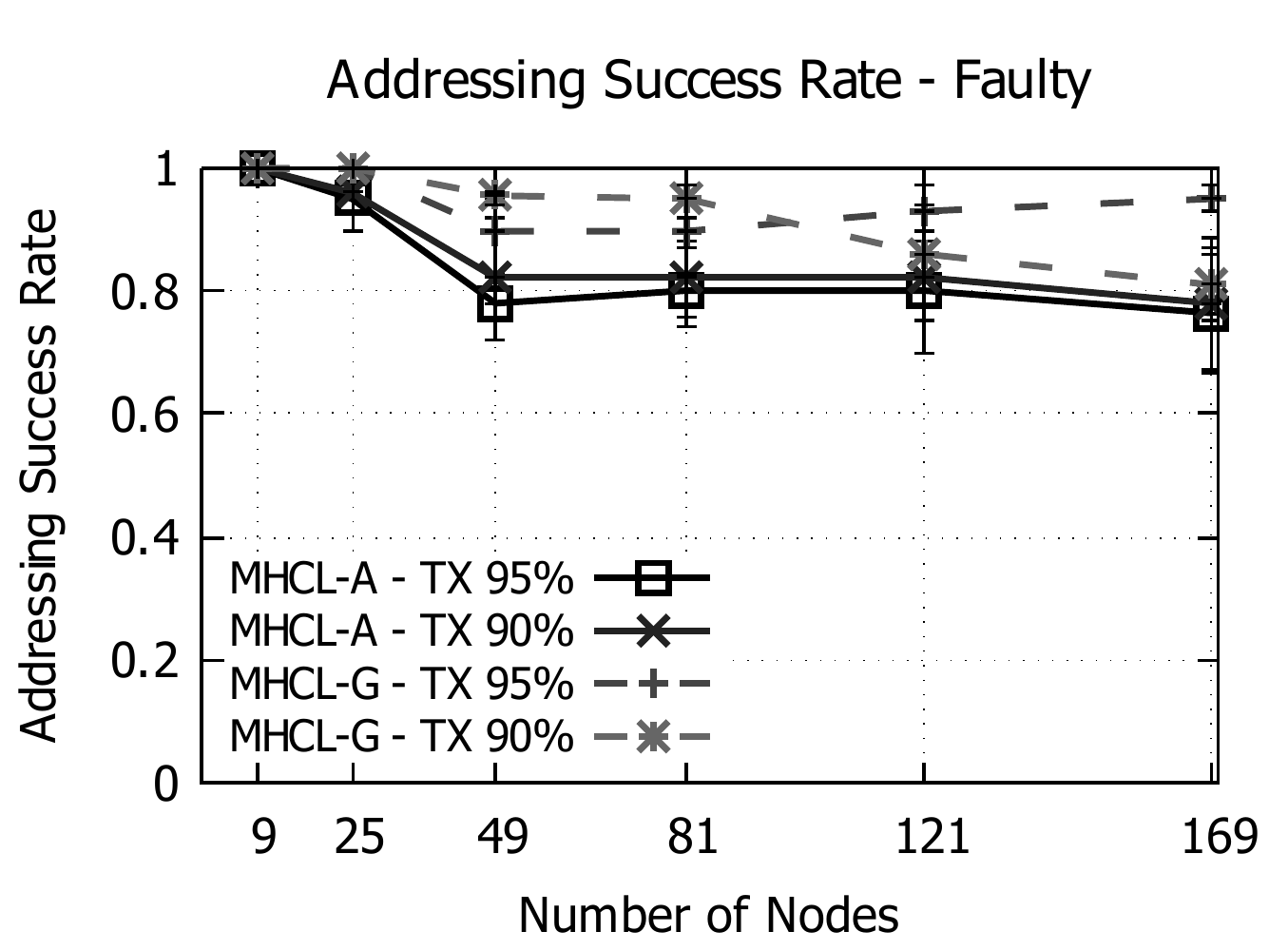}}
  \subfigure[RX Failure]
            {\includegraphics[width=.45\columnwidth]{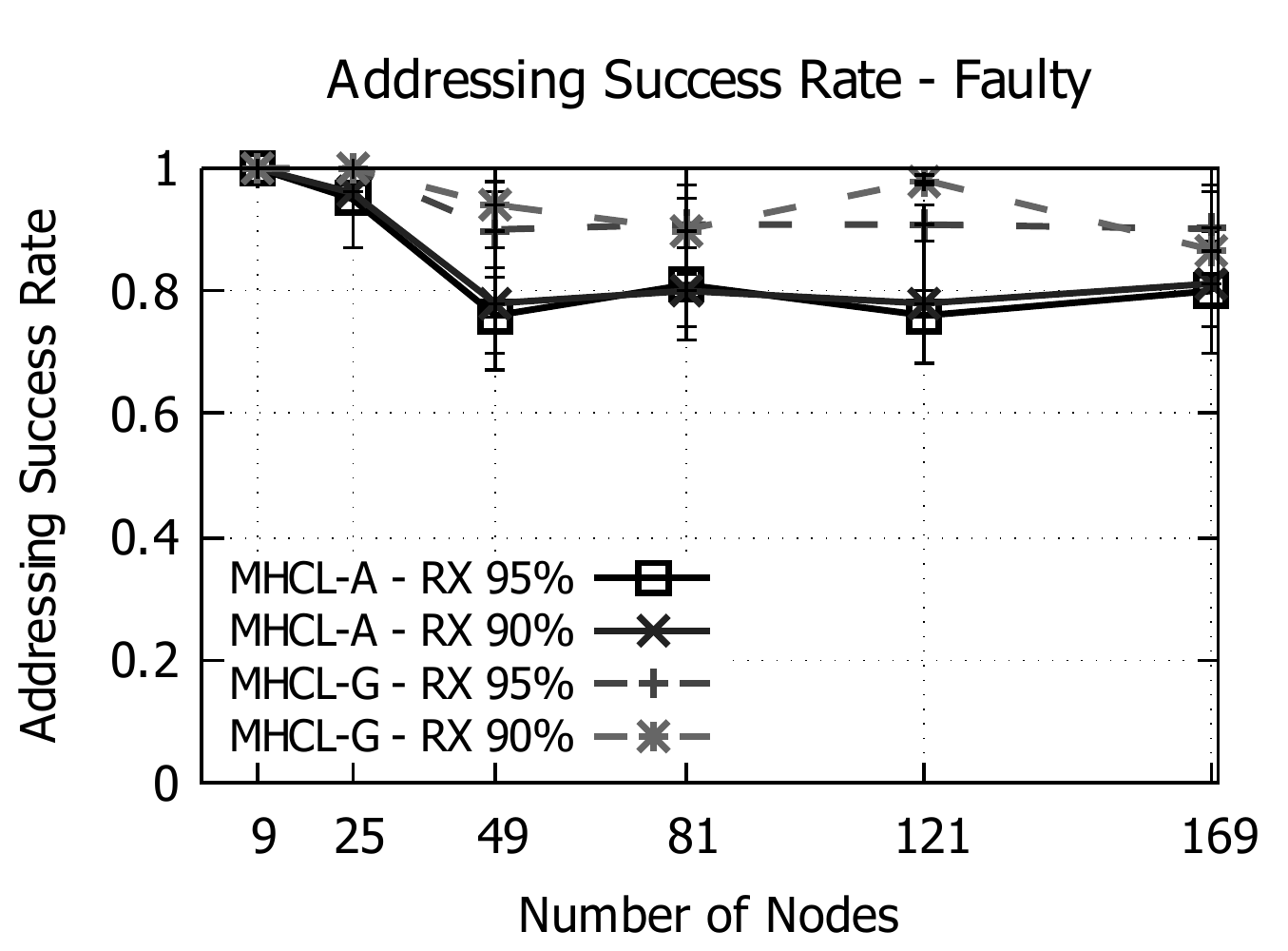}}
            \caption{Addressing rate - Faulty - Uniform Topology.}\label{fig:txend_f_u}
\end{center}
\end{figure}

\subsection{Application Layer Routing}\label{subsec:appLayer}

Finally, we look at the delivery rate of top-down application
messages (Figures 
\ref{fig:txdsc} , \ref{fig:txdsc_f}, and \ref{fig:txdsc_f_u}). Note that the simulated application layer does not perform retransmission of messages and acknowledgments,
therefore, collisions and failures of links and nodes can cause permanent loss of messages. In the case of upward messages, as the routing
algorithm has not changed, the strategies have similar values. 
In the downward, or top-down, routing, the performance difference between MHCL
and RPL is striking. In the 169-node scenario, for instance, MHCL delivers
almost $4\times$ more messages than RPL, on average. This occurs because, in RPL, the nodes do not have enough memory to store the
routing tables needed to address all their descendants, thereby preventing the routing of all
messages, leading to a low downward routing rate. Even though MHCL delivers
significantly more messages relatively to RPL, the overall routing success rate
can be as low as 50\% in some scenarios (RPL's success rate is as low as 10\% in
these cases), even without link or node failures.
This can be explained by high collision rates between concurrent messages and
lower than 100\% addressing success rate of MHCL. Note that this ratio could be improved by
collision-avoidance techniques and additional retransmissions of $DIO_{MHCL}$
messages.

When we compare message delivery success, including address allocation messages,
between regular and uniform topologies, we can see that uniform has slightly
better results than regular. This can be explained by the fact that the height
of the DAG constructed by RPL is smaller in the uniform topology, due to greater
distance variability among nodes. For example, when $n=121$, in the regular
case, $height(DAG)=20$ but, in the uniform case, $height(DAG)=7$ on average.
Given that failure and collision probabilities accumulate over multiple hops
of the path followed by each message, the overall message loss is higher in the
regular network.

\begin{figure}[t]
\centering
\includegraphics[width=0.6\columnwidth]{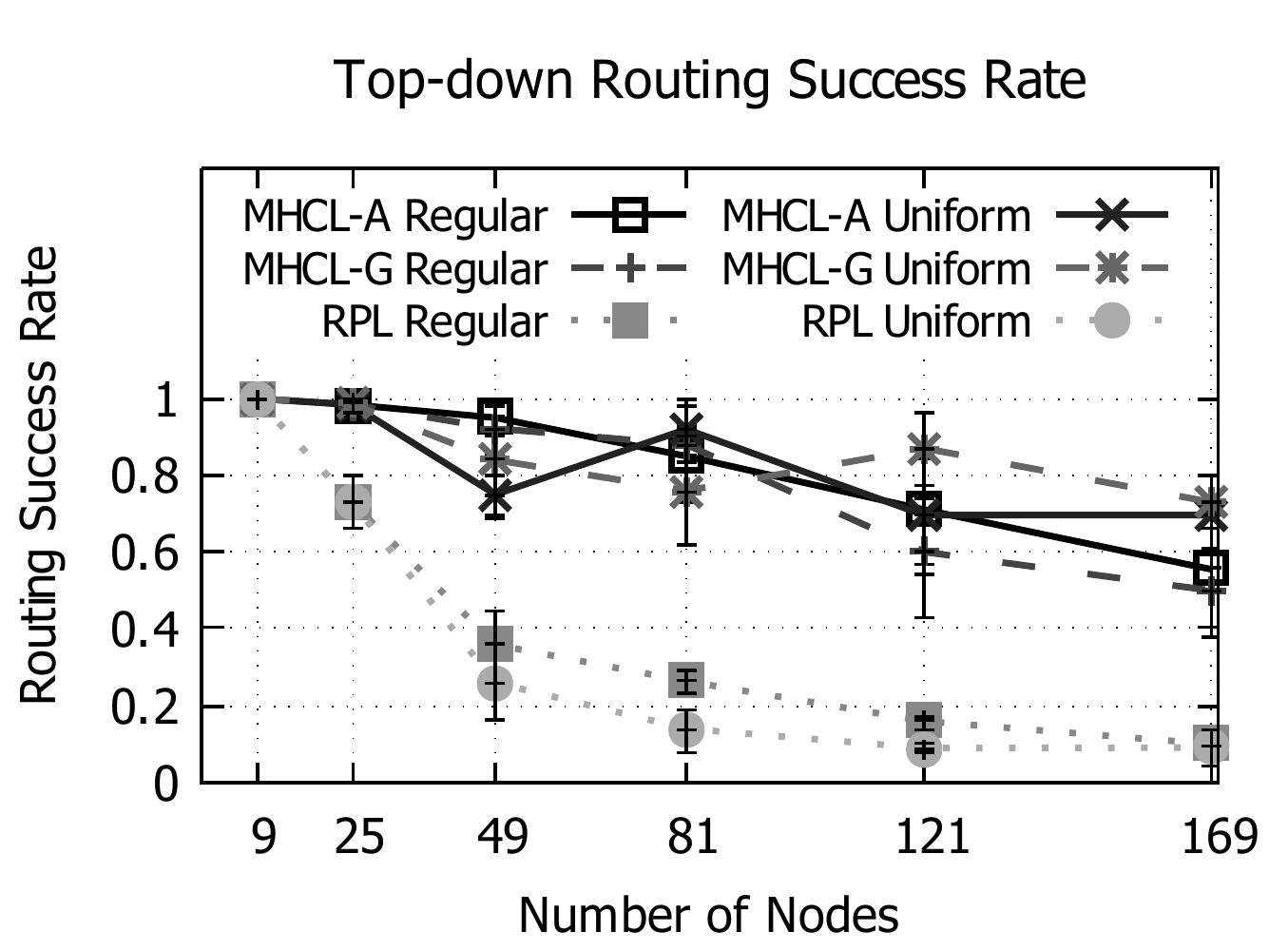}
\caption{Top-down Routing.}\label{fig:txdsc}
\end{figure}

\begin{figure}[h]
\begin{center}
  \subfigure[TX Failure]
  {\includegraphics[width=.45\columnwidth]{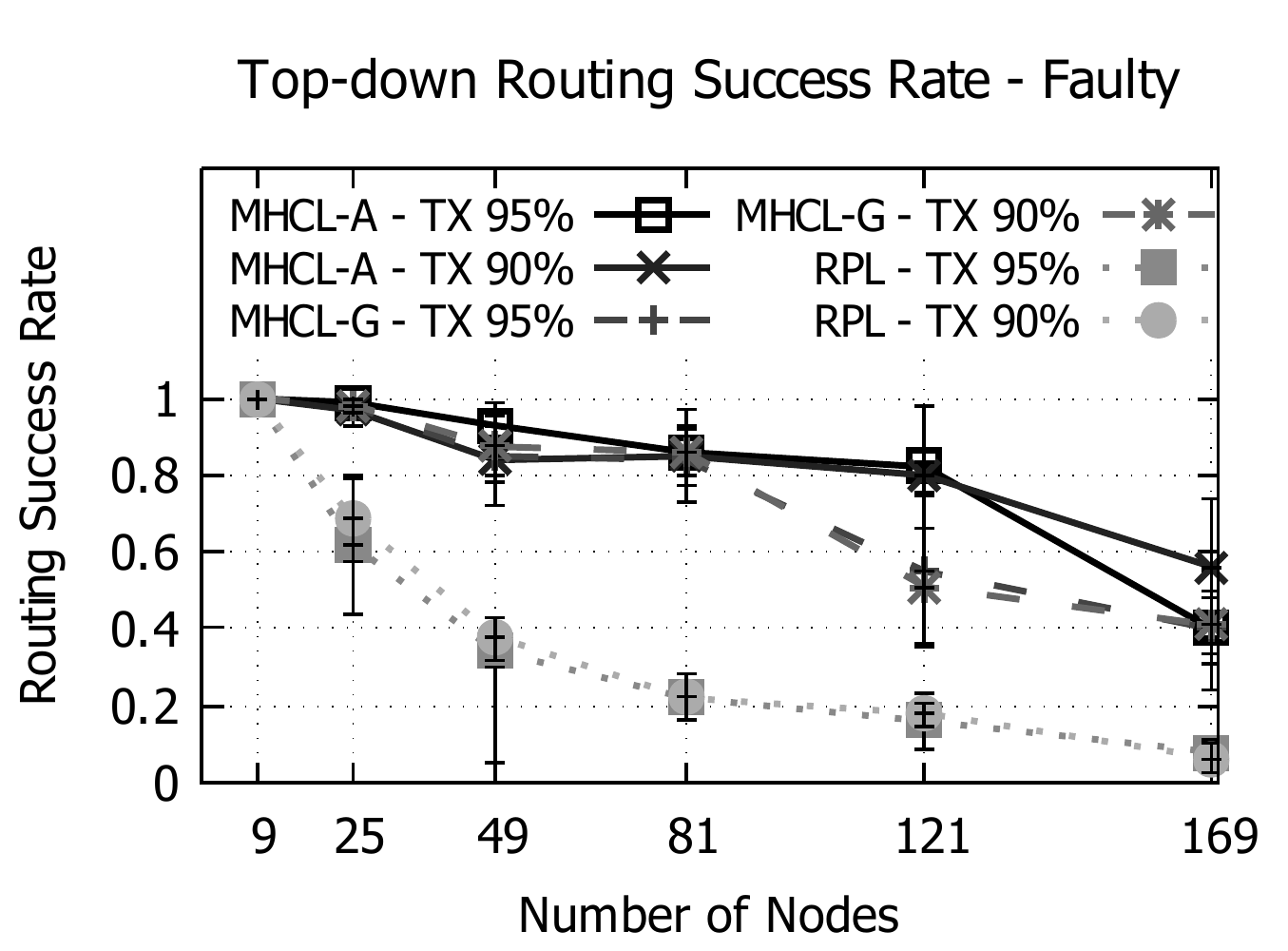}}
  \subfigure[RX Failure]
            {\includegraphics[width=.45\columnwidth]{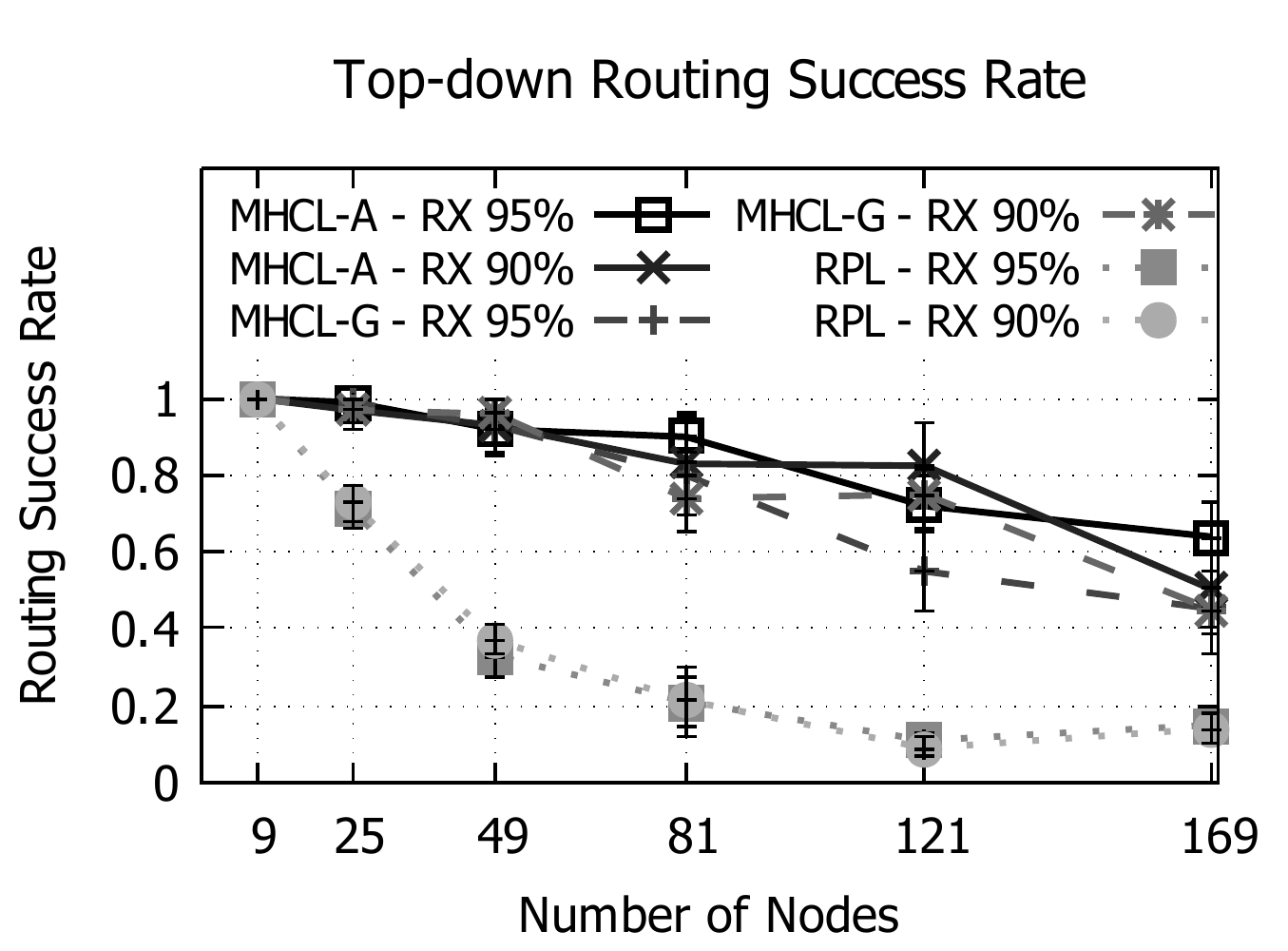}}
            \caption{Top-down Routing - Faulty - Regular Topology.}\label{fig:txdsc_f}
\end{center}
\end{figure}

\begin{figure}[h]
\begin{center}
  \subfigure[TX Failure]
  {\includegraphics[width=.45\columnwidth]{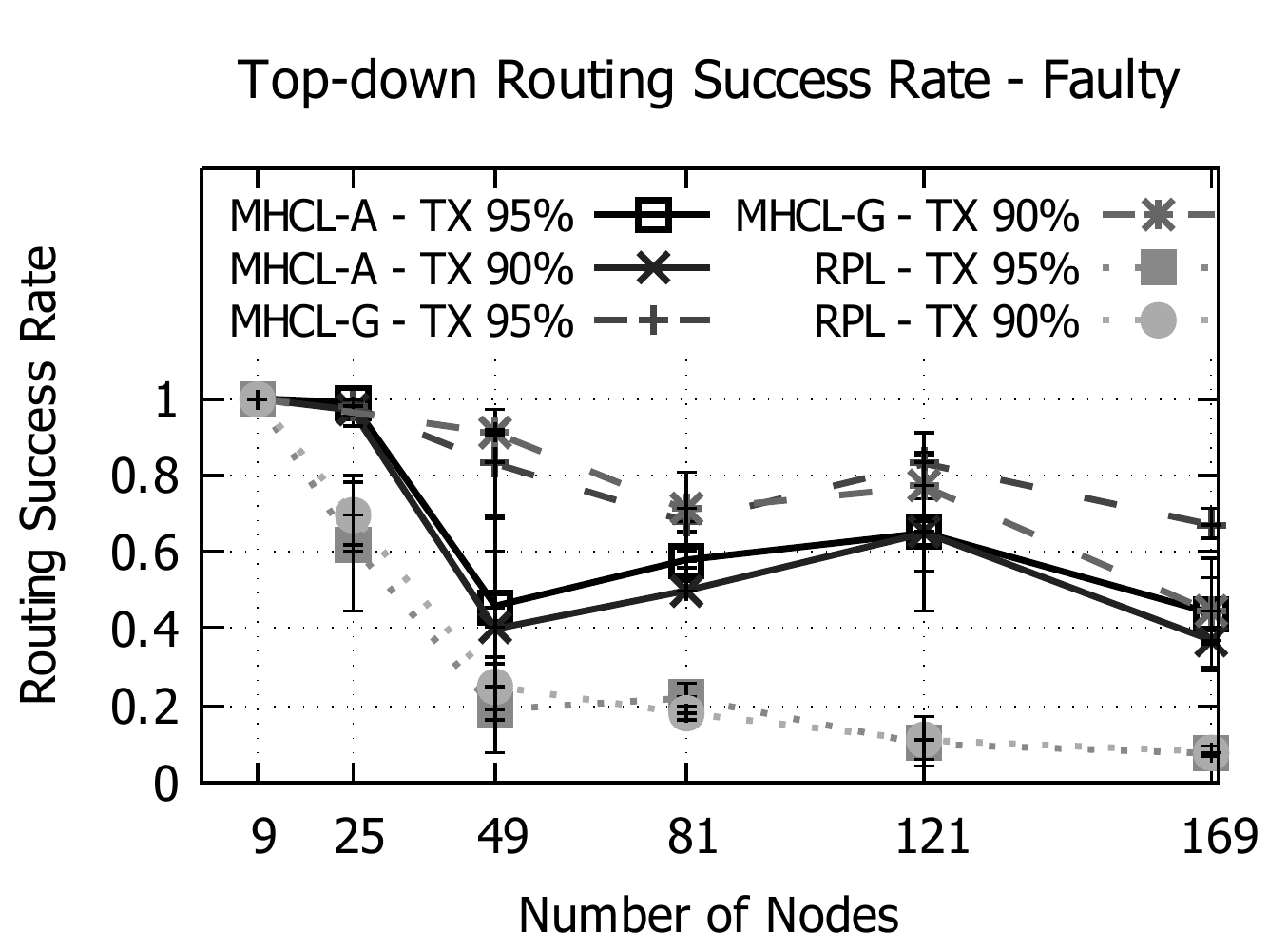}}
  \subfigure[RX Failure]
            {\includegraphics[width=.45\columnwidth]{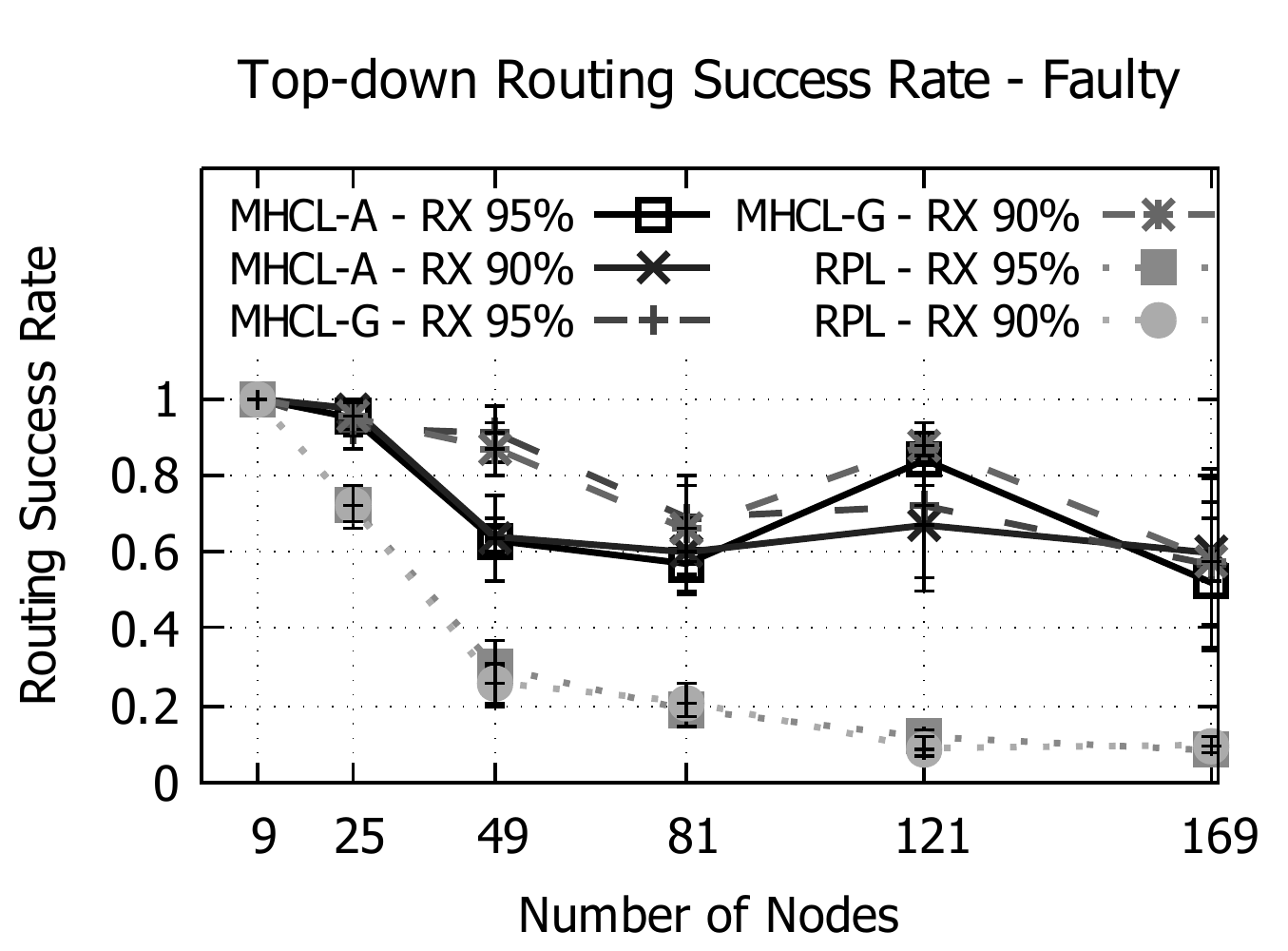}}
            \caption{Top-down Routing - Faulty - Uniform Topology.}\label{fig:txdsc_f_u}
\end{center}
\end{figure}
\section{Related Work}
\label{sec:related}

Standard routing protocols for
Low Power Lossy Networks, such as CTP (Collection Tree
Protocol~\cite{Fonseca:2009}) and RPL (IPv6 Routing Protocol for Low-Power and Lossy Networks~\cite{rfc6550}), are
designed to optimize bottom-up (many-to-one) data flows, by maintaining a tree-based structure. The advantage of a tree topology,
besides loop-avoidance, is a constant-size (i.e., independent of the network size)
routing table, since each node just needs to establish who his parent-node is
and maintain only that information for packet forwarding. The disadvantage of
such structures is that other communication patterns, like top-down
(one-to-many) or one-to-one bidirectional data flows, are not
easily implemented. 
The problem with one-directional routing is it makes it
infeasible to build several useful network functions, such as configuration
routines and reliable mechanisms to ensure the delivery of data end-to-end.
In order to do that, addition communication routines have to be implemented and
extensive routing information has to be inserted into the routing
tables of memory-constrained nodes.

Some works have adressed this problem from different perspectives~\cite{Rein12,
Duque13, xctp, mhclSBRC}. \cite{Rein12} presents CBFR, a routing scheme
that builds upon collection protocols to enable point-to-point communication.
Each node in the collection tree stores the addresses of its direct and indirect
child nodes using Bloom filters to save memory
space.
\cite{Duque13} presents ORPL, which also uses bloom filters and brings opportunistic routing to RPL to decrease control traffic overload. Both
protocols suffer from false positives problem, that arises from the use of Bloom
filters. Furthermore, \cite{xctp} presents XCTP, an extension of CTP, which uses
opportunistic and reverse-path routing to enable bi-directional communication in
CTP. XCTP is efficient in terms of message overload, but exhibits the problem of
high memory footprint, since each node needs to store an entry in the local
routing table for every data flow going through that node.

As opposed to reactive strategies, such as ORPL and XCTP, MHCL consists of a
proactive approach, which is based on a
hierarchical address allocation scheme, aimed at reducing routing table sizes
needed for top-down data flows. In~\cite{mhclSBRC} we discussed a few
preliminary tests with an early implementation of MHCL in smaller and more
restricted network scenarios.

In \cite{Pan08} the so-called long-thin (LT) network topology is presented. In
this kind of topology, a network may have a number of linear paths of nodes as backbones connecting to each other. From real
experiments, they observe that such topology is quite general in many
applications and deployments. Given that, they analyze how the address
assignment strategy and tree routing scheme defined in ZigBee fails in this
topology. To solve this problem, they propose a new address assignment and
routing scheme for LT WSN. \cite{Rein12} presents CBFR, a routing scheme
that builds upon collection protocols to enable point-to-point communication. To
do that, each node in the collection tree stores the addresses of its direct and
indirect child nodes. Since memory is a scarce resource, this information is
stored in Bloom filters, a space-efficient data structure.
Finally, \cite{Duque13} presents ORPL. ORPL brings opportunistic routing to RPL,
aiming for low-latency, reliable communication in duty-cycled networks. To route
upwards, at the MAC layer, ORPL uses anycast over a low-power-listening MAC.
Since the opportunistic routing uses anycast, nodes do not need to choose a next
hop and therefore do not need a traditional routing table. The authors introduce
the notion of routing set, the sub-DODAG rooted at the node. ORPL supports
any-to-any traffic by first routing upwards to any common ancestor, and then
downwards to the destination, as RPL downward mode. ORPL uses the same data
structure to store route information as CBFR, a bloom filter. However, in ORPL,
instead of choosing a next hop, nodes anycast packets. When receiving a packet,
nodes decide whether to forward it or not, and send a link-layer acknowledgment only if they choose to act as next hop.

There are several studies evaluating the specification and operation of the RPL protocol. In~\cite{Ko2012} it was pointed out how
challenging it is to make the two modes of downward routing, storing and
non-storing, operate simultaneously. To solve interoperability problems, the
authors propose that nodes operating in storing mode should be able to work with source routing
headers, used in non-storing mode, and nodes operating in non-storing
mode should send hop-by-hop route notification messages, rather than end-to-end. Other challenges of RPL were exposed in
\cite{Clausen2011}, such as lack of specification of some parts of the protocol.
For example, while DIO messages use Trickle to specify the timing of messages,
details about DAO type messages are obscure and no timer for these messages has
been defined. 

In \cite{RPLtinyos} and \cite{RPLContiki} the RPL protocol was
evaluated according to its implementation in TinyOS and ContikiOS, respectively.
In \cite{RPLtinyos}, 
an implementation of RPL is presented, called TinyRPL, which uses the Berkeley Low-power IP
stack (BLIP) for TinyOS \cite{Levis2005}. The TinyRPL implementation has all the
basic mechanisms presented in RPL definition, while omitting all optional ones. For validation, the RPL
protocol was compared with the Collection Tree Protocol (CTP)
\cite{Fonseca:2009}.
One of the benefits of RPL, compared with the CTP, is the possibility of various types of traffic (P2P, P2MP
and MP2P), and its ability to connect nodes to the Internet directly, exchanging
packets with global IPv6 addresses. In \cite{RPLContiki} the performance of RPL
protocol in the operating system Contiki OS~\cite{Dunkels:2004} was studied. The
authors evaluated both the setup and the data collection phases of the protocol,
giving insights into metrics, such as signaling overhead, latency and energy
consumption.

\section{Conclusions}\label{sec:conclusion}

In this paper we proposed MHCL: IPv6 Multihop Host Configuration for Low-Power
Wireless Networks.
The main advantages of MHCL are:
(1) Memory efficiency: using the cycle free topology of the network to perform a
hierarchical partitioning of the address space available at a border router, the
addresses of nodes belonging to the same sub-tree are grouped into a single
entry in the node's routing table, making the size of the table $O(k)$ where $k$
is the number of the node's direct children; this is opposed to size $O(n)$ of
the RPL's downward routing tables, where $n$ is the total number of descendants
of each node; (2) Time efficiency: MHCL's main routines are based on timers that rapidly adapt
to network dynamics, enabling a fast address configuration phase; when compared to RPL, the
network setup time was statistically equivalent; (3) Number of control messages:
the number of messages sent by MHCL is controlled by Trickle-based timers,
which avoid flooding the network when network topology is
stable; in comparison with RPL, the number of DIO messages was statistically the
same and that of DAO messages was significantly lower; (4) Efficiency in the delivery of application messages:
downward message delivery success of MHCL was significantly higher
than that of the RPL protocol in all simulated scenarios.

\begin{table}[htbp]
\caption{Stabilization parameter values} 
\label{tab:stabilizationParams}
 \centering
\small
 \begin{tabular}{@{}llllll@{}}
\toprule
  Parameter & spChild & spParent & spLeaf & spRoot & $DIO_{min}$\\
	\midrule
  Value & 2 & 4 & 4 & 8 & 6 \\
	\bottomrule
 \end{tabular}
 \end{table}

\end{document}